\documentclass[a4paper,USenglish]{lipics-v2018}
\usepackage{microtype}
\nolinenumbers 
\usepackage{array}
\usepackage{times}
\usepackage{graphicx}
\usepackage{amsmath}
\usepackage{amssymb}
\usepackage{amsfonts}
\usepackage{cite}
\usepackage{multido}
\usepackage{algorithm}
\usepackage{algorithmic}
\usepackage{tikz}
\usepackage{xcolor}
\usepackage{multirow}
\usepackage{paralist}
\usepackage{tabularx}

\newcolumntype{L}[1]{>{\raggedright\arraybackslash}p{#1}}
\newcolumntype{C}[1]{>{\centering\arraybackslash}p{#1}}
\newcolumntype{R}[1]{>{\raggedleft\arraybackslash}p{#1}}

\usetikzlibrary{shadows,patterns,shapes,arrows,decorations.pathmorphing,backgrounds,positioning,fit,plotmarks}

 \def\myendproof{{\ \vbox{\hrule\hbox{%
   \vrule height1.3ex\hskip0.8ex\vrule}\hrule }}\par}

 \newenvironment{proofAppendix}[1]{\noindent{\bf Proof of #1. }}{\hfill\myendproof}

\newcommand{\floor}[1]{\left\lfloor{#1}\right\rfloor}
\newcommand{\ceiling}[1]{\left\lceil{#1}\right\rceil}
\newcommand{\setof}[1]{\left\{{#1}\right\}}

\newcommand{\resourceArbM}{3-\frac{1}{M}}

\tikzset{
    task/.style={shade, shading=radial, rectangle,minimum height=.1cm,
        inner color=#1!20, outer color=#1!60!gray},
    task1/.style={task=yellow, minimum width=13mm},
    task2/.style={task=orange, minimum width=13mm},
    task3/.style={task=red, minimum width=13mm},
    task4/.style={task=green, minimum width=13mm},
    task5/.style={task=blue, minimum width=13mm},
    task6/.style={task=purple, minimum width=13mm},
    task7/.style={task=cyan, minimum width=13mm},
    task8/.style={task=pink, minimum width=13mm},
}

\usepackage{etoolbox}
\providebool{techreport}
\setbool{techreport}{true}

\floatname{algorithm}{Algorithm}
\newcommand{\algorithmicinput}{\textbf{Input:}}
\newcommand{\INPUT}{\item[\algorithmicinput]}

\tikzstyle{entity} = [top color=white, bottom color=red!40,
                            draw=blue!50!black!100, drop shadow]
\tikzstyle{relationship} = [diamond, shape aspect=2, top color=white, bottom color=red!20,
                                  draw=red!50!black!100, drop shadow]
\tikzstyle{end} = [top color=white, bottom color=blue!30,
                            draw=blue!50!black!100, drop shadow]

\tikzstyle{max} = [top color=white, bottom color=orange!60,
                            draw=blue!50!black!100, drop shadow]

\title{Packing Sporadic Real-Time Tasks on Identical
  Multiprocessor Systems}
\titlerunning{Packing Sporadic Real-Time Tasks on Identical
  Multiprocessor Systems}

\author{Jian-Jia Chen}{Department of Computer Science, TU Dortmund University,
Germany}{jian-jia.chen@cs.uni-dortmund.de}{0000-0001-8114-9760}{}
\author{Nikhil Bansal}{Eindhoven University of Technology, the Netherlands}{n.bansal@tue.nl}{}{} 
\author{Samarjit Chakraborty}{Technical University of Munich (TUM), Germany}{samarjit@tum.de}{0000-0002-0503-6235}{} 
\author{Georg von der Br\"uggen}{Department of Computer Science, TU Dortmund
University, Germany}{georg.von-der-brueggen@tu-dortmund.de}{0000-0002-8137-3612}{}

\authorrunning{J.-J. Chen, N. Bansal, S. Chakraborty, G. von der
Br\"uggen}

\Copyright{Jina-Jia Chen, Nikhil Bansal, Samarjit Chakraborty, and Georg von der
Br\"uggen}

\keywords{multiprocessor partitioned scheduling, approximation factors}

\subjclass{\ccsdesc[500]{Computer systems organization~Real-time systems}}

\EventEditors{Wen-Lian Hsu, Der-Tsai Lee, and Chung-Shou Liao}
\EventNoEds{3}
\EventLongTitle{29th International Symposium on Algorithms and Computation (ISAAC 2018)}
\EventShortTitle{ISAAC 2018}
\EventAcronym{ISAAC}
\EventYear{2018}
\EventDate{December 16--19, 2018}
\EventLocation{Jiaoxi, Yilan, Taiwan}
\EventLogo{}
\SeriesVolume{123}
\ifbool{techreport}{
\ArticleNo{XX} 
}{
\ArticleNo{<117>} 
}

\begin{document}

\maketitle

\begin{abstract}
In real-time systems, in addition to the functional correctness recurrent tasks must fulfill timing constraints to ensure the correct behavior of the system. Partitioned scheduling is widely used in real-time systems, i.e., the tasks are statically assigned onto processors while ensuring that all  timing constraints are met. The decision version of the problem, which is to check whether the deadline constraints of tasks can be satisfied on a given number of identical processors, has been known ${\cal NP}$-complete in the strong sense. Several studies on this problem are based on approximations involving resource augmentation, i.e.,  speeding up individual processors.  This paper studies another type of resource augmentation by allocating additional processors, a  topic that has not been explored until recently.  We provide polynomial-time algorithms and analysis, in which the approximation factors are dependent upon the input instances. Specifically, the factors are related to the maximum ratio of the period to the relative deadline of a task in the given task set. We also show that these algorithms unfortunately cannot achieve a constant approximation factor for general cases. Furthermore, we prove that the problem does not admit any asymptotic polynomial-time approximation scheme (APTAS) unless ${\cal P}={\cal NP}$ when the task set has constrained deadlines, i.e., the relative deadline of a task is no more than the period of the task.
\end{abstract}

\section{Introduction}
\label{sec:intro}

The sporadic task model has been widely adopted to model recurring
executions of tasks in real-time systems \cite{Mok:1983:FDP:888951}.  A
sporadic real-time task $\tau_i$ is defined with a \emph{minimum inter-arrival
time} $T_i$, its timing constraint or \emph{relative deadline} $D_i$, and its
(worst-case)
\emph{execution time} $C_i$. A sporadic task represents an infinite
sequence of task instances, also called \emph{jobs}, that arrive with the
minimum inter-arrival time constraint. That is, any two consecutive jobs of task
$\tau_i$ should be
temporally separated by at least $T_i$.  When a job of task $\tau_i$
arrives at time $t$, the job must finish no later than its
\emph{absolute deadline} $t+D_i$.  According to the Liu and Layland
task model \cite{liu73scheduling}, the minimum inter-arrival time of a
task can also be interpreted as the \emph{period} of the task.

To schedule real-time tasks on multiprocessor platforms, there have
been three widely adopted paradigms: partitioned, global, and
semi-partitioned scheduling. A comprehensive survey of
multiprocessor scheduling in real-time systems can be found in
\cite{DBLP:journals/csur/DavisB11}. In this paper, we consider
\emph{partitioned scheduling}, in which tasks are statically
partitioned onto processors. This means that all the jobs of a task are
executed on a specific processor, which reduces the online scheduling
overhead since each processor can schedule the sporadic tasks
assigned on it without considering the tasks on the other
processors. Moreover, we consider preemptive scheduling on each
processor, i.e, a job may be preempted by another job on the
processor. For scheduling sporadic tasks on one processor, the
(preemptive) earliest-deadline-first (EDF) policy is optimal
\cite{liu73scheduling} in terms of meeting timing
constraints, in the sense that if the task set is schedulable then it will also
be schedulable under EDF. In EDF, the job (in the ready queue) with the
earliest absolute deadline has the highest priority for execution.
Alternatively, another widely adopted scheduling paradigm is
(preemptive) fixed-priority (FP) scheduling, where all jobs
released by a sporadic task have the same priority level.  

The complexity of testing whether a task set can be feasibly
scheduled on a uniprocessor depends on the relations
between the relative deadlines and the minimum inter-arrival times of
tasks. An input task set is said to have (1)~\emph{implicit
  deadlines} if the relative deadlines of sporadic tasks are equal to their
minimum inter-arrival times, (2)~\emph{constrained deadlines} if the minimum
inter-arrival times are no less than their relative deadlines, and
(3)~\emph{arbitrary deadlines}, otherwise.

On a uniprocessor, checking the feasibility for an implicit-deadline task set is
simple and well-known: the timing constraints are met by EDF if and only if the total
utilization $\sum_{\tau_i \in {\bf T}} \frac{C_i}{T_i}$ is at most $100\%$
\cite{liu73scheduling}.
Moreover, if every task $\tau_i$ on the processor is with $D_i \geq T_i$, it is
not difficult to see that testing whether the total utilization is less than or
equal to $100\%$ is also a necessary and sufficient schedulability test. This
  can be achieved by considering a more stringent case which sets $D_i$ to $T_i$
  for every $\tau_i$.  
  Hence, this special case of
  arbitrary-deadline task sets can be reformulated to task sets with implicit
  deadlines 
  without any loss of precision.
However, determining the schedulability for task sets with constrained or
arbitrary deadlines in general is much harder, due to the complex interactions
between the deadlines and the periods, and in particular is known to be
co${\mathcal NP}$-hard or co${\mathcal
NP}$-complete~\cite{DBLP:conf/soda/EisenbrandR10,DBLP:conf/ecrts/Ekberg015,DBLP:conf/rtss/Ekberg015}.


In this paper, we consider partitioned scheduling in homogeneous multiprocessor
systems.
Deciding if an implicit-deadline task set is schedulable on multiple processors
is already ${\cal NP}$-complete in the strong sense under partitioned
scheduling.
To cope with these ${\cal NP}$-hardness issues, one natural approach is to focus
on approximation algorithms, i.e., polynomial time algorithms that produce an
approximate solution instead of an exact one.
In our setting, this translates to designing algorithms that can find a feasible
schedule using either (i) faster or (ii) additional processors.
The goal, of course, is to design an algorithm that uses the least speeding up
or as few additional processors as possible.
In general, this approach is referred to as resource augmentation and is used
extensively to analyze and compare scheduling algorithms. See for example
\cite{PTS04} for a survey and motivation on why this is a useful measure for
evaluating the quality of scheduling algorithms in practice. However, such a
measure also has its potential pitfalls as recently studied and reported by Chen
et al. \cite{chen2017pitfalls}.
Interestingly, 
it turns out that there is a huge difference regarding the approximation factors
depending on whether it is possible to increase the processor speed or the
number of processors.
As already discussed in~\cite{ChenECRTS12}, approximation by speeding up is known as the \emph{multiprocessor
partitioned scheduling problem}, and by allocating more processors is known as
the \emph{multiprocessor partitioned packing problem}.
We study the latter one in this paper.

Formally, an algorithm ${\cal A}$ for the multiprocessor partitioned packing
problem is said to have an approximation factor $\rho$, if given any task set
${\bf T}$, it can find a feasible partition of ${\bf T}$ on $\rho M^*$
processors, where $M^*$ is the minimum (optimal) number of processors required
to schedule ${\bf T}$.
However, it turns out that the approximation factor is not the best measure in
our setting (it is not fine-grained enough).
For example, it is ${\cal NP}$-complete to decide if an implicit-deadline task
set is schedulable on 2 processors or whether 3 processors are necessary.
Assuming $\cal{P}\neq \cal{NP}$, this rules out the possibility of any efficient
algorithm with approximation factor better than $3/2$, as shown in~\cite{ChenECRTS12}.
(This lower bound 
is further lifted to $2$ for sporadic tasks in
Section~\ref{sec:hardness}.) The problem with this example is that it does
not rule out the possibility of an algorithm that only needs $M^*+1$ processors.
Clearly, such an algorithm is almost as good as optimum when $M^*$ is large and
would be very desirable.\footnote{Indeed, there are (very ingenious) algorithms
known for the implicit-deadline partitioning problem that use only $M^* +
O(\log^2 M^*)$ processors \cite{KK82}, based on the connection to the
bin-packing problem.} To get around this issue, a more refined measure is the
so-called asymptotic approximation factor.
An algorithm ${\cal A}$ has an \emph{asymptotic} approximation factor $\rho$ if
we can find a schedule using at most $\rho M^*+\alpha$ processors, where
$\alpha$ is a constant that does not depend on $M^*$. An algorithm is called an
asymptotic polynomial-time approximation scheme (APTAS) if, given an arbitrary
accuracy parameter $\epsilon>0$ as input, it finds a schedule using
$(1+\epsilon)M^* +O(1)$ processors and its running time is polynomial assuming
$\epsilon$ is a fixed constant.

For implicit-deadline task sets, the multiprocessor partitioned
scheduling problem, by speeding up, is equivalent to the Makespan problem
\cite{DBLP:journals/siamam/Graham69}, and the multiprocessor
partitioned packing problem, by allocating more processors, is
equivalent to the bin packing problem \cite{b:Gary79}.  The
Makespan problem admits polynomial-time approximation schemes (PTASes), by
Hochbaum and Shmoys \cite{DBLP:journals/jacm/HochbaumS87}, and the bin
packing problem admits asymptotic polynomial-time approximation
schemes (APTASes), by de la Vega and Lueker
\cite{DBLP:journals/combinatorica/VegaL81,KK82}.

When considering sporadic task sets with constrained or arbitrary
deadlines, the problem becomes more complicated.  When adopting
speeding-up for resource augmentation, the deadline-monotonic
partitioning proposed by Baruah and Fisher
\cite{DBLP:conf/rtss/BaruahF05,DBLP:journals/tc/BaruahF06} has been
shown to have a $\resourceArbM$ speed-up factor in~\cite{Chakraborty2011a},
where $M$ is the given number of identical processors.
The studies in~\cite{BaruahRTSS2011,ChenECRTS12,DBLP:conf/latin/BansalRSVZ14}
provide polynomial-time approximation schemes for some special cases when
speeding-up is possible. The PTAS by Baruah~\cite{BaruahRTSS2011} requires that
$\frac{D_{\max}}{D_{\min}}, \frac{C_{\max}}{C_{\min}},
\frac{T_{\max}}{T_{\min}}$ are constants, where $D_{\max}$ ($C_{\max}$ and
$T_{\max}$, respectively) is the maximum relative deadline (worst-case execution
time and period, respectively) in the task set and $D_{\min}$ ($C_{\min}$ and
$T_{\min}$, respectively) is the minimum relative deadline (worst-case execution
time and period, respectively) in the task set.
It was later shown  in~\cite{ChenECRTS12,DBLP:conf/latin/BansalRSVZ14} that the
complexity only depends on $\frac{D_{\max}}{D_{\min}}$.
If $\frac{D_{\max}}{D_{\min}}$ is a constant, there exists a PTAS developed by
Chen and Chakraborty~\cite{ChenECRTS12}, which admits feasible task partitioning
by speeding up the processors by $(1+\epsilon)$. The approach in
\cite{ChenECRTS12} deals with the multiprocessor partitioned scheduling problem
as a vector scheduling problem~\cite{DBLP:journals/siamcomp/ChekuriK04} by
constructing (roughly) $(1/\epsilon)\log \frac{D_{\max}}{D_{\min}}$ dimensions
and then applies the PTAS of the vector scheduling problem developed by Chekuri
and Khanna \cite{DBLP:journals/siamcomp/ChekuriK04} in a black-box manner. 
Bansal et al. \cite{DBLP:conf/latin/BansalRSVZ14} exploit the special structure
of the vectors and give a faster vector scheduling algorithm that is a
quasi-polynomial-time approximation scheme (qPTAS) even if
$\frac{D_{\max}}{D_{\min}}$ is polynomially bounded.

However,  augmentation by allocating additional processors, i.e.,
the multiprocessor partitioned packing problem, 
has not been explored until recently in real-time systems. Our previous
work in~\cite{ChenECRTS12} has initiated the study for minimizing
the number of processors 
for real-time tasks. While~\cite{ChenECRTS12} mostly focuses on approximation
algorithms for resource augmentation via speeding up, it also showed that
for the multiprocessor partitioned packing problem there does not exist any
APTAS for
arbitrary-deadline task sets, unless ${\cal P} = {\cal NP}$.
However, the proof in~\cite{ChenECRTS12} for the non-existence of
APTAS only works
when the input task set ${\bf T}$ has \emph{exactly} two types of
tasks in which one type consists of tasks with relative deadline less
than or equal to its period (i.e., $D_i \leq T_i$ for some $\tau_i$ in ${\bf T}$)
and another type consists of tasks with
relative deadline larger than its period (i.e., $D_j > T_j$ for some $\tau_j$ in ${\bf T}$).
Therefore, it cannot be directly applied for constrained-deadline task sets.

For the results, from the literature and also this paper, related to
the multiprocessor partitioned scheduling and packing problems,
Table~\ref{tab:summary} provides a short summary.

\begin{table}[t]
  \scalebox{0.622}{
  \begin{tabular}{|C{3.4cm}|C{3.5cm}|C{3.8cm}|C{4.2cm}|C{5.2cm}|}
    \hline
    & implicit deadlines& constrained deadlines& arbitrary deadlines & arbitrary deadlines (dependent on $\frac{D_{\max}}{D_{\min}}$) \\
    \hline
     partitioned EDF  & PTAS \cite{DBLP:journals/jacm/HochbaumS87}& $2.6322$-speed up \cite{Chakraborty2011a}& $3$-speed up \cite{Chakraborty2011a}& PTAS \cite{ChenECRTS12} for constant  $\frac{D_{\max}}{D_{\min}}$\\
     scheduling&&&&qPTAS \cite{DBLP:conf/latin/BansalRSVZ14} for polynomial
     $\frac{D_{\max}}{D_{\min}}$\\
    \hline 
     partitioned FP  & 
$\frac{7}{4}$
    \cite{DBLP:journals/tc/BurchardLOS95},  $1.5$
    \cite{DBLP:conf/waoa/KarrenbauerR10} &  $2.84306$
    speed-up \cite{ChenECRTS2016-Partition} & $3$-speed
    up\cite{ChenECRTS2016-Partition} &\\
    scheduling& (extended from packing)&&&\\ 
    \hline 
    \multirow{2}{*}{partitioned packing} &
    \multirow{2}{*}{APTAS \cite{DBLP:journals/combinatorica/VegaL81}} & 
    non-existence of APTAS$^\sharp$ & \multicolumn{2}{C{9.4cm}|}{non-existence
    of APTAS \cite{ChenECRTS12}}\\\cline{3-5} & &   
    \multicolumn{3}{C{14.2cm}|}{$2\lambda$-approximation$^\sharp$, asymptotic $\frac{2}{1-\gamma}$-approximation$^\sharp$, non-existence of $(2-\epsilon)$-approximation$^\sharp$}\\
      \hline
  \end{tabular}}
\caption{Summary of the
multiprocessor partitioned scheduling and packing problems, unless ${\cal P}={\cal NP}$, where
  $\gamma=\max_{\tau_i \in {\bf T}}\frac{C_i}{\min\{T_i, D_i\}}$,
  $\lambda=\max_{\tau_i \in {\bf T}}\max\{\frac{T_i}{D_i},1\}$, and
  $D_{\max}$ ($D_{\min}$
  ) is the task set's maximum (minimum
  ) relative deadline. 
  A $^\sharp$ marks results from
  this paper.
  }
  \label{tab:summary}
\end{table}

{\bf Our Contributions} This paper studies the multiprocessor partitioned
packing problem in much more detail.
On the positive side, when the ratio of the period of a constrained-deadline
task to the relative deadline of the task is at most $\lambda=\max_{\tau_i \in {\bf
    T}}\max\{\frac{T_i}{D_i},1\}$, in Section~\ref{sec:bin-packing},
we provide a simple polynomial-time algorithm with a $2\lambda$-approximation
factor.
In Section~\ref{sec:deadline-monotonic}, we show that  the deadline-monotonic
partitioning algorithm in
\cite{DBLP:conf/rtss/BaruahF05,DBLP:journals/tc/BaruahF06}
has an asymptotic  $\frac{2}{1-\gamma}$-approximation factor for the packing
problem,  where $\gamma=\max_{\tau_i \in
  {\bf T}}\frac{C_i}{\min\{T_i, D_i\}}$.
In particular, when $\gamma$ and $\lambda$ are not constant, adopting the
worst-fit or best-fit strategy in the deadline-monotonic partitioning algorithm
is shown to have an $\Omega(N)$ approximation factor, where $N$ is the number of
tasks.  In contrast, from \cite{Chakraborty2011a}, it is known that both
strategies have a speed-up factor $3$, \emph{when the resource augmentation is
to speed up processors}.
We also show that speeding up processors can be much more powerful than
allocating more processors. Specifically, in Section~\ref{sec:hardness}, we
provide input instances, in which the only feasible schedule is to run each task
on an individual processor but the system requires only one processor with a
speed-up factor of $(1+\epsilon)$, where $0 < \epsilon < 1$.

On the negative side, in Section \ref{sec:noAPTAS}, we show that there does not
exist any asymptotic polynomial-time approximation scheme (APTAS) for the
multiprocessor partitioned packing problem for task sets with constrained
deadlines, unless ${\cal P}={\cal NP}$.
As there is already an APTAS for the implicit deadline case, this together with
the result in~\cite{ChenECRTS12} gives a complete picture of the approximability
of multiprocessor partitioned packing for different types of task sets, as shown
in Table~\ref{tab:summary}.



\section{System Model}
\label{sec:model}


\subsection{Task and Platform Model}

We consider a set ${\bf T}=\setof{\tau_1, \tau_2, \ldots, \tau_N}$ of
$N$ independent sporadic real-time tasks. Each of these tasks releases an
infinite number of task instances, called jobs. A task $\tau_i$ is defined by
$(C_i, T_i, D_i)$, where
$D_i$ is its relative
deadline, $T_i$ is its minimum inter-arrival time (period), and $C_i$
is its (worst-case) execution time. For a job released at time $t$, the next
job must be released no earlier than $t+T_i$ and it must finish (up to) $C_i$
amount of execution before the jobs absolute deadline at $t+D_i$.
The \emph{utilization} of task $\tau_i$ is denoted by $u_i=\frac{C_i}{T_i}$. 
We consider platforms with identical processors,
i.e., the execution and timing property remains no matter which processor a task
is assigned to. According to
the relations of the relative deadlines and the minimum inter-arrival
times of the tasks in ${\bf T}$, the task set can be identified to be
with (1) implicit deadlines, i.e., $D_i=T_i~\forall \tau_i$, (2) constrained
deadlines, i.e., $D_i\leq T_i~\forall \tau_i$, or (3) arbitrary deadlines,
otherwise.
The cardinality of a set ${\bf X}$  
is denoted by $|{\bf X}|$.

In this paper we focus on partitioned scheduling, i.e., each task is statically
assigned to a fixed processor and all jobs of the task is executed on the assigned processor. On
each processor, the jobs related to the tasks allocated to that processor are
scheduled using preemptive earliest deadline first (EDF) scheduling. This means
that at each point the job with the shortest absolute deadline is executed,
and if a new job with a shorter absolute deadline arrives the currently executed
job is preempted and the new arriving job starts executing. A task set 
can be feasibly scheduled by EDF (or EDF is a feasible schedule) on a processor
if the timing constraints can be fulfilled by using EDF. 

\subsection{Problem Definition}

Given a task set ${\bf T}$, a \emph{feasible task partition} on $M$
identical processors is a collection of $M$ subsets, denoted ${\bf T}_1,
{\bf T}_2, \ldots, {\bf T}_M$, 
 such that
\begin{compactitem}
\item ${\bf T}_j \cap {\bf T}_{j'}=\emptyset$ for all $j\neq j'$,
\item $\cup_{j=1}^{M} {\bf T}_j$ is equal to the input task set ${\bf
    T}$, and
\item set ${\bf T}_j$ can meet the timing constraints 
by EDF  scheduling on a processor $j$.
\end{compactitem}

\begin{definition}
{\it The multiprocessor partitioned packing problem:} The objective is
to find a feasible task partition on $M$ identical processors with the
minimum $M$.\end{definition}


We assume that  $u_i \leq 100\%$ and
$\frac{C_i}{D_i} \leq 100\%$ for any task $\tau_i$ since otherwise 
there cannot be a feasible partition. 

\subsection{Demand Bound Function}


This paper focuses on the case where the arrival times of the sporadic
tasks are not specified, i.e., they arrive according to their interarrival
constraint and not according to a pre-defined pattern.
Baruah et al. \cite{baruah_sporadic} have shown that in this case the
worst-case pattern is 
to release the first job of tasks
synchronously (say, at time $0$ for notational brevity), 
and all subsequent jobs as early as possible.
Therefore, as shown in 
\cite{baruah_sporadic}, the \emph{demand
  bound function} ${\sc dbf}(\tau_i, t)$ of a task $\tau_i$ 
  that specifies the maximum demand of task $\tau_i$ to be released and
finished  within any time
interval with length $t$ is defined as
\begin{equation}
  \label{eq:dbf}
  {\sc dbf}(\tau_i, t) = \max\left\{0,
  \floor{\frac{t-D_i}{T_i}}+1\right\}\times C_i.
\end{equation}
The exact schedulability test of EDF, to verify whether EDF can feasibly
schedule the given task set on a processor, is to check whether the
summation of the demand bound functions of all the tasks is always
less than $t$ for all $t \geq 0$ \cite{baruah_sporadic}.

\section{Reduction to Bin Packing}
\label{sec:bin-packing}

When considering tasks with
implicit deadlines, the multiprocessor partitioned packing problem is
equivalent to the bin packing problem~\cite{b:Gary79}. Therefore, even though
the packing becomes more complicated when considering tasks with arbitrary
or constrained deadlines, it is pretty straightforward to handle the
problem by using existing algorithms for the bin packing problem if
the maximum ratio $\lambda$ of the period to the relative
deadline among the tasks, i.e., $\lambda=\max_{\tau_i \in {\bf
    T}}\max\{\frac{T_i}{D_i}, 1\}$, is not too
large.

For a given task set ${\bf T}$, we can basically transform the input
instance to a related task instance ${\bf T}^\dagger$ by creating task
$\tau_i^\dagger$ based on task $\tau_i$ in ${\bf T}$ such that
\begin{compactitem}
\item  $T_i^\dagger$ is $D_i$, $C_i^\dagger$ is $C_i$, and $D_i^\dagger$ is
$D_i$ when $T_i \geq D_i$ for $\tau_i$, and
\item $D_i^\dagger$ is $T_i^\dagger$, $C_i^\dagger$ is $C_i$ and
$T_i^\dagger$ is $T_i$ when $T_i < D_i$ for $\tau_i$.
\end{compactitem}











Now, we can adopt any
greedy fitting algorithms (i.e., a task is assigned to ``one''
allocated processor that is feasible; otherwise, a new processor is
allocated and the task is assigned to the newly allocated processor)
for the bin packing problem by considering only the utilization of
transformed tasks in ${\bf T}^\dagger$ for the multiprocessor
partitioned packing problem, as presented in~\cite[Chapter
8]{b:Vijay}. The construction of ${\bf T}^{\dagger}$ 
has a time complexity of $O(N)$, 
 and the greedy fitting algorithm has a time complexity of $O(NM)$. 


\begin{theorem}
  \label{thm:gamma-small-3}
  Any greedy
  fitting algorithm by considering ${\bf T}^\dagger$ for task
  assignment is a $2\lambda$-approximation algorithm for the
  multiprocessor partitioned packing problem.
\end{theorem}
\begin{proof}
  Clearly, as we only reduce the relative deadline and the periods,
  the timing parameters in ${\bf T}^\dagger$ are more stringent than
  in~${\bf T}$. Hence, a feasible task partition for ${\bf T}^\dagger$ on
  $M$ processors also yields a corresponding feasible task
  partition for ${\bf T}$ on $M$ processors.  As ${\bf
    T}^\dagger$ has implicit deadlines, we know that any task
  subset in ${\bf T}^\dagger$ with total utilization no more than
  $100\%$ can be feasibly scheduled by EDF on a processor, 
  and therefore the original tasks in that subset as well. For any greedy
  fitting algorithms that use $M$ processors, using the same proof as in
  \cite[Chapter 8]{b:Vijay}, we get  $\sum_{\tau_i \in {\bf T}^\dagger}
  \frac{C_i^\dagger}{T^\dagger_i}  >  \frac{M}{2}$.

  By definition, we know that
      $\sum_{\tau_i \in {\bf T}} \frac{C_i}{T_i} \geq
    \sum_{\tau_i^\dagger \in {\bf T}^\dagger} \frac{C_i^\dagger}{\lambda T^\dagger_i}  > \frac{M}{2\lambda}$.
  Therefore, any feasible solution for ${\bf T}$ uses at least
  $\frac{M}{2\lambda}$ processors and the approximation factor is hence proved.
\end{proof}

\section{Deadline-Monotonic Partitioning under EDF Scheduling}
\label{sec:deadline-monotonic}

This section presents the worst-case analysis of the
deadline-monotonic partitioning strategy, proposed by Baruah and
Fisher \cite{DBLP:journals/tc/BaruahF06,DBLP:conf/rtss/BaruahF05}, for
the multiprocessor partitioned packing problem. Note that the underlying 
scheduling algorithm is EDF but the tasks are considered in the
\textit{deadline-monotonic} (DM) order. Hence, in this section, we index
the tasks accordingly from the shortest relative deadline to the longest, i.e., $D_i \leq D_j$ if $i
< j$.  Specifically, 
in the DM 
partitioning,  the approximate demand bound function ${\sc dbf}^*(\tau_i, t)$ 
is used to approximate 
Eq.~\eqref{eq:dbf}, where
\begin{equation}
  \label{eq:dbf-approx}
    {\sc dbf}^*(\tau_i, t) =\left\{
    \begin{array}{ll}
      0 & \mbox{if } t < D_i\\
     \left(\frac{t-D_i}{T_i}+1\right)C_i& \mbox{otherwise.}
  \end{array}
\right.
\end{equation}
Even though the 
DM partitioning algorithm in
\cite{DBLP:journals/tc/BaruahF06,DBLP:conf/rtss/BaruahF05} is designed
for the multiprocessor partitioned scheduling problem, it can be
easily adapted to deal with the multiprocessor partitioned packing
problem. For completeness, we revise the algorithm in
\cite{DBLP:journals/tc/BaruahF06,DBLP:conf/rtss/BaruahF05} for the
multiprocessor partitioned packing problem and present the  
pseudo-code in Algorithm~\ref{alg:dmp}.
As discussed in
\cite{DBLP:journals/tc/BaruahF06,DBLP:conf/rtss/BaruahF05}, when a task $\tau_i$ is considered, a processor $m$ among
the allocated processors where both the following conditions hold
\begin{align}
&C_i + \sum_{\tau_j \in {\bf T}_m} {\sc dbf}^*(\tau_j, D_i) \leq D_i   \label{eq:demand-bound-test}\\
&u_i+ \sum_{\tau_j \in {\bf T}_m} u_j \leq 1   \label{eq:utilization-bound-test}
\end{align}
is selected to assign task $\tau_i$, where ${\bf T}_m$ is the set of
the tasks (as a subset of $\setof{\tau_1, \tau_2, \ldots, \tau_{i-1}}$), which have been assigned to
processor $m$ before considering $\tau_i$. If there is no $m$ where 
both Eq.~\eqref{eq:demand-bound-test} and
Eq.~\eqref{eq:utilization-bound-test} hold, a new processor is allocated and
task $\tau_i$ is assigned to the new processor.
The order in which the already allocated processors are considered depends on
the fitting strategy:
\begin{compactitem}
\item \emph{first-fit} (FF) strategy: choosing  the feasible $m$ with the
minimum index;
\item \emph{best-fit} (BF) strategy: choosing, among the feasible processors,
 $m$ with the maximum approximate demand bound at time $D_i$;
\item \emph{worst-fit} (WF) strategy: choosing $m$ with the
minimum approximate demand bound at time $D_i$.
\end{compactitem}

\begin{algorithm}[t]
  \caption{Deadline-Monotonic Partitioning}
  \label{alg:dmp}
  \begin{algorithmic}[1]
    \INPUT set ${\bf T}$ of $N$ tasks;

    \STATE re-index (sort) tasks such that $D_i \leq D_j$ for $i < j$;

    \STATE  $M \leftarrow 1$, ${\bf T}_1 \leftarrow \setof{\tau_1}$;

    \FOR {$i=2$ to $N$}

    \IF {$\exists m \in \setof{1, 2, \ldots, M}$ such that both
      (\ref{eq:demand-bound-test}) and (\ref{eq:utilization-bound-test})
       hold} \label{algo:fitting-test}

    \STATE choose $m \in \setof{1, 2, \ldots, M}$ \emph{\bf by preference} such that both
      (\ref{eq:demand-bound-test}) and (\ref{eq:utilization-bound-test})
       hold;

    \STATE assign $\tau_i$ to processor $m$ with ${\bf T}_m \leftarrow {\bf T}_m
    \cup \setof{\tau_i}$;

    \ELSE
    \STATE $M \leftarrow M+1$; ${\bf T}_M \leftarrow \setof{\tau_i}$;
    \ENDIF

    \ENDFOR

    \STATE return feasible task partition ${\bf T}_1, {\bf T}_2,
    \ldots, {\bf T}_M$;
  \end{algorithmic} \label{algo:fitting}
\end{algorithm}

For a given number of processors, it has been proved in~\cite{Chakraborty2011a}
that the speed-up factor of the DM 
partitioning 
is at most $3$, 
independent from the fitting strategy.
However, if the objective is to minimize the number of
allocated processors, we will show that DM 
partitioning has an approximation factor of at least $\frac{N}{4}$ (in the worst
case) when the best-fit or  worst-fit strategy is adopted. We will prove this by
explicitly constructing two concrete task sets with this property. Afterwards,
we show that the asymptotic approximation factor of DM 
partitioning is at
most $\frac{2}{1-\gamma}$ for packing, where  
$\gamma=\max_{\tau_i \in {\bf T}}\frac{C_i}{\min\{T_i, D_i\}}$.






\begin{theorem}
  \label{thm:best-fit}
  The approximation factor of the deadline-monotonic partitioning
  algorithm with the best-fit strategy is at least $\frac{N}{4}$ when
  $N \geq 8$ and the schedulability test is based on
  Eq.~\eqref{eq:demand-bound-test} and
Eq.~\eqref{eq:utilization-bound-test}.
\end{theorem}
\begin{proof}
The theorem is proven by 
providing a task set that can be scheduled on two processors 
  but where Algorithm~\ref{algo:fitting} when applying the best-fit strategy uses
  $\frac{N}{2}$ processors. Under the assumption that $K \geq 4$ is an
  integer, $N$ is $2K$, and $H$ is sufficiently large, i.e., $H \gg K^K$, such
  a task set can be constructed as:
  \begin{compactitem}
  \item Let $D_1=1$, $C_1=1/K$, and $T_1=H$.
  \item For  $i=2, 4, \ldots,2K$, let $D_i = K^{\frac{i}{2}-1}$,
    $C_i = K^{\frac{i}{2}-2}$, and $T_i=D_i$.
  \item For $i=3, 5, \ldots,2K-1$, let $D_i =
    K^{\frac{i-1}{2}}$, $C_i = K^{\frac{i-1}{2}} -
    K^{\frac{i-1}{2}-1}$, and $T_i=H$.
  \end{compactitem}
  The task set can be scheduled on two processors under EDF if all tasks with an
  odd index are assigned to processor 1 and all tasks with an even index are
  assigned to processor 2. On the other hand, the best-fit strategy assigns
  $\tau_i$ to processor $\ceiling{\frac{i}{2}}$. The resulting solution uses
  $K$ processors.  Details
  are in the Appendix\ifbool{techreport}{.}{ in~\cite{packing-arxiv-2018}.}
\end{proof}

\begin{theorem}
  \label{thm:worst-fit}
  The approximation factor of the deadline-monotonic partitioning
  algorithm with the worst-fit strategy is at least $\frac{N}{4}$  when
  the schedulability test is based on
  Eq.~\eqref{eq:demand-bound-test}~and~Eq.~\eqref{eq:utilization-bound-test}.
\end{theorem}

\begin{proof}
The proof is very similar to the proof of Theorem~\ref{thm:best-fit},
considering the task set:
  \begin{compactitem}
  \item Let $D_1=1$, $C_1=1$, and $T_1=H$.
  \item For  $i=2, 4, \ldots,2K$, let $D_i = K^{\frac{i}{2}}$,
    $C_i = K^{\frac{i}{2}-1}$, and $T_i=D_i$.
  \item For $i=3, 5, \ldots,2K-1$, let $D_i =
    K^{\frac{i-1}{2}}$, $C_i = K^{\frac{i-1}{2}} -
    K^{\frac{i-1}{2}-1}$, and $T_i=H$.
  \end{compactitem}
  Odd tasks are assigned to processor 1 and even tasks to processor 2
  the task set is schedulable while $\tau_i$ is assigned to processor
  $\ceiling{\frac{i}{2}}$ using the worst-fit strategy. Details are in
  the Appendix\ifbool{techreport}{.}{ in~\cite{packing-arxiv-2018}.}
\end{proof}

\begin{theorem}
  \label{thm:gamma-small-2}
  The DM partitioning algorithm is an
  asymptotic $\frac{2}{1-\gamma}$-approximation algorithm for the multiprocessor
  partitioned packing problem, when  $\gamma=\max_{\tau_i \in {\bf
      T}}\frac{C_i}{\min\{T_i, D_i\}}$ and $\gamma < 1$.
\end{theorem}

\begin{proof}
We consider the task $\tau_l$ which is the task that is responsible for the
last processor that is allocated by Algorithm~\ref{alg:dmp}. The other
processors are categorized
  into two disjoint sets ${\bf M}_1$ and ${\bf M}_2$, depending on whether  
  Eq.~\eqref{eq:demand-bound-test}  or
  Eq.~\eqref{eq:utilization-bound-test} is violated when
  Algorithm~\ref{alg:dmp} tries to assign $\tau_l$ (if both conditions are
  violated, the processor is in ${\bf M}_1$). The two sets are considered
  individually and the maximum number of processors in both sets is determined
  based on the minimum utilization for each of the processors.
  Afterwards, a necessary condition for the amount of processors that is
  at least needed for a feasible solution is provided  and the relation between
  the two values proves the theorem. Details can be found in the Appendix\ifbool{techreport}{.}{ in~\cite{packing-arxiv-2018}.} 
\end{proof}

\section{Hardness of Approximations}
\label{sec:hardness}

It has been shown in~\cite{ChenECRTS12,BaruahRTSS2011} that a PTAS
exists for augmenting the resources by speeding up. A straightforward
question is to see whether such PTASes will be helpful for bounding the
lower or upper bounds for multiprocessor partitioned
packing. Unfortunately, the following theorem shows that using speeding up to get a lower bound for the
number of required processors is not useful.

\begin{theorem}
  \label{theorem:gap-speedup-vs-processors}
  There exists a set of input instances, in which the number of
  allocated processors is up to $N$, while the task set can be
  feasibly scheduled by EDF with a speed-up factor $(1+\epsilon)$
  on a processor, where $0 < \epsilon < 1$.
\end{theorem}
\begin{proof}
  We provide a set of input instances, with the property described in
  the statement:
  \begin{compactitem}
  \item Let $D_1=1$, $C_1=1$, and $T_1=\frac{(1+\epsilon)^{N-2}}{\epsilon^{N-1}}$.
  \item For any $i=2, 3,\ldots,N$, let $D_i = \frac{(1+\epsilon)^{i-2}}{\epsilon^{i-1}}$,
    $C_i = D_i$, and $T_i=\frac{(1+\epsilon)^{N-2}}{\epsilon^{N-1}}$.
  \end{compactitem}
  Since $C_i=D_i$ for any task $\tau_i$, assigning any two tasks on the same
  processor is infeasible without speeding up. Therefore, the only
  feasible processor allocation is $N$ processors and to assign
 each task individually on one processor.  
However, by speeding up the system by a factor $1+\epsilon$, the
  tasks can be feasibly scheduled on one processor due to
  $\sum_{i=1}^{N} \frac{dbf(\tau_i, t)}{1+\epsilon} \leq t$ for any $t
  > 0$. A proof is in the Appendix\ifbool{techreport}{.}{ in~\cite{packing-arxiv-2018}.}  Hence, the gap between these two
  types of resource augmentation is up to $N$.
\end{proof}




Moreover, the following theorem shows the inapproximability for a
factor $2$ without adopting asymptotic approximation.

\begin{theorem}
  \label{thm:packing-lowerbound-2}
  For any $\epsilon > 0$, there is no polynomial-time approximation
  algorithm with an approximation factor of $2-\epsilon$ for
  the multiprocessor partitioned packing problem, unless ${\cal
    P}={\cal NP}$.
\end{theorem}
\begin{proof}
  Suppose that there exists such a polynomial-time algorithm ${\cal A}$ with
  approximation factor $2-\epsilon$. 
   $\cal{A}$ can be used to decide if a task set ${\bf T}$ is schedulable on a uniprocessor,
  which would contradict the $co{\cal NP}$-hardness \cite{DBLP:conf/soda/EisenbrandR10} of this problem.
  Indeed, we simply run $\cal{A}$ on the input instance.
  If $\cal{A}$ returns a feasible schedule using one processor, we already have a uniprocessor schedule. On the other hand,
  if $\cal{A}$ requires at least two processors, then we know that any optimum solution needs $\geq \ceiling{\frac{2}{2-\epsilon}} =
  2$ processors, implying that the task set ${\bf T}$ is not schedulable
  on a uniprocessor.
\end{proof}


\section{Non-Existence of APTAS  for Constrained Deadlines}
\label{sec:noAPTAS}
We now show that there is no APTAS when considering constrained-deadline task sets, unless
${\cal P}={\cal NP}$. 
The proof is based on an \mbox{L-reduction} (informally an approximation
preserving reduction) from a special case of the \emph{vector packing problem},
i.e., the 2D dominated vector packing problem. 

\subsection{The 2D Dominated Vector Packing Problem}

 The \emph{vector packing problem} is defined as follows:
\begin{definition}
{\it The vector packing problem:} Given a set ${\bf V}$ of vectors
$[v_1, v_2, \ldots, v_N]$ with $d$ dimensions, in which $1 \geq v_{i,j} \geq 0$ is
the value for vector $v_i$ in the $j$-th dimension, the problem is
to partition ${\bf V}$ into $M$ parts ${\bf V}_1,\ldots,{\bf V}_M$ such that $M$ is minimized and
each part ${\bf V}_m$ is feasible in the sense that  $\sum_{v_i\in {\bf V}_m} v_{i,j}\leq 1$ for all $1 \leq j \leq d$. That is, for each dimension $j$,
the sum of the $j$-th coordinates of the vectors in ${\bf V}_m$ is at most $1$. \hfill\myendproof
\end{definition}

We say that a subset ${\bf V}'$ of ${\bf V}$ can be \emph{feasibly
  packed in a bin} if $\sum_{v_i \in {\bf V}'} v_{i,j} \leq 1$ for all
$j$-th dimensions.
Note that for $d=1$ this is precisely the bin-packing problem.
The vector packing problem does not admit any
polynomial-time asymptotic approximation scheme even in the case of $d=2$ dimensions, unless ${\cal P}={\cal
  NP}$ \cite{DBLP:journals/ipl/Woeginger97}.

Specifically, the proof in~\cite{ChenECRTS12} for the non-existence of
APTAS for task sets with arbitrary deadlines comes from an L-reduction
from the $2$-dimensional vector packing problem as follows: For a
vector $v_i$ in ${\bf V}$, a task $\tau_i$ is created with 
$D_i=1$, $C_i = v_{i,2}$, and $T_i=\frac{v_{i,2}}{v_{i,1}}$.  However, a trivial
extension from \cite{ChenECRTS12} to constrained deadlines 
does not work, since for the transformation of the task set
we need to
assume that $v_{i,1} \leq v_{i,2}$ for any $v_i \in {\bf V}$ so that
$T_i \geq 1 = D_i$ for every reduced task $\tau_i$. This becomes
problematic, as one dimension in the vectors in such input
instances for the two-dimensional vector packing problem can be
totally ignored, and the input instance becomes a special case equivalent to the traditional bin-packing problem,
which admits an APTAS.
We will show that the hardness is equivalent to a special case
of the two-dimensional vector packing problem, called the \emph{
  two-dimensional dominated vector packing} (2D-DVP) problem, in
Section~\ref{sec:2d-dvp-versus-packing}. 
\begin{definition}
The\emph{
  two-dimensional dominated vector packing} (2D-DVP) problem is a
special case of the two-dimensional vector packing problem with 
following conditions for each vector $v_i \in {\bf V}$:
\begin{compactitem}
\item $v_{i,1} > 0$, and
\item if $v_{i,2} \neq 0$, then $v_{i,1}$ is dominated by $v_{i,2}$,
  i.e., $v_{i,2} > v_{i,1}$.
\end{compactitem}
Moreover, we further assume that $v_{i,1}$ and $v_{i,2}$ are rational numbers for every $v_i \in {\bf V}$. \hfill\myendproof
\end{definition}


Here, some tasks are created with implicit 
deadlines (based on vector $v_i$ if $v_{i,2}$ is $0$) and some 
tasks with strictly constrained deadlines (based on vector $v_i$ if
$v_{i,2}$ is not $0$).   
 However, the 2D-DVP problem is a special case of the
two-dimensional vector packing problem, and the implication for
$v_{i,2} > v_{i,1}$ when $v_{i,2} \neq 0$ does not hold in the proof in
\cite{DBLP:journals/ipl/Woeginger97}.
We note, that the proof for the non-existence of an APTAS
  for the two-dimensional vector packing problem in
  \cite{DBLP:journals/ipl/Woeginger97} is erroneous. However,
  the result still holds. Details are in the Appendix\ifbool{techreport}{.}{ in~\cite{packing-arxiv-2018}.}
Therefore,
we will provide a proper $L$-reduction
in Section~\ref{sec:2d-dvp-hard} to show the non-existence of
APTAS for the multiprocessor partitioned packing problem for tasks
with constrained deadlines.

\subsection{2D-DVP Problem and Packing Problem}
\label{sec:2d-dvp-versus-packing}


 We now show that the packing problem is at least as hard as the 2D-DVP problem
 from a complexity point of view.
For vector $v_i$ with $v_{i,2} > v_{i,1}$, we create a
corresponding task $\tau_i$ with
\[
D_i = 1,~~ C_i = v_{i,2},~~ T_i = \frac{v_{i,2}}{v_{i,1}}.
\]
Clearly, $D_i < T_i$ for such tasks. Let $H$ be a \emph{common multiple}, not
necessary the least, of the periods $T_i$ of 
the tasks  constructed above.
By the assumption that all the values in the 2D-DVP problem are rational numbers
and $v_{i,1} > 0$ for every vector $v_i$, we know that $H$ exists and can be
calculated in $O(N)$. 
 For vector $v_i$ with $v_{i,2} = 0$, we create a corresponding
implicit-deadline task $\tau_i$ with 
\[ 
T_i = D_i = H, ~~C_i = v_{i,1}T_i.
\]
 The following lemma shows the related schedulability condition.

\begin{lemma}
  \label{lemma:common-deadline}
  Suppose that the set ${\bf T}_m$ of tasks assigned on a processor
  consists of (1) strictly constrained-deadline tasks, denoted by 
  ${\bf T}_m^{<}$, with a common relative deadline $1=D$ and
  (2) implicit-deadline tasks, i.e., ${\bf T}_m \setminus{\bf T}_m^{<}$,
  in which the period is a common integer multiple $H$ of the periods of the
  strictly constrained-deadline tasks.  EDF schedule is feasible for the set
  ${\bf T}_m$ of tasks on a processor if and only if
  \[
   \sum_{\tau_i \in {\bf T}_m^{<}} C_i \leq 1 \mbox{ and }    \sum_{\tau_i \in {\bf T}_m} u_i \leq 1.
   \]
\end{lemma}
\begin{proof}
  {\bf Only if}: This 
  is straightforward as the task set 
  cannot meet the timing constraint when $\sum_{\tau_i \in {\bf
      T}_m^{<}} \frac{C_i}{D} > 1$ or $\sum_{\tau_i \in {\bf T}_m} u_i
  > 1$.

  {\bf If}: If $\sum_{\tau_i \in {\bf T}_m^{<}} \frac{C_i}{D} \leq 1$ and
  $\sum_{\tau_i \in {\bf T}_m} u_i \leq 1$, we know that when $t < D$, then 
  $\sum_{\tau_i \in {\bf T}_m}{\sc dbf}(\tau_i, t) = 0$.
  When $D\leq t < H$, we have
\begin{align}
    \label{eq:dbf-sum-common-deadline}
    &\sum_{\tau_i \in {\bf T}_m}{\sc dbf}(\tau_i, t) =  \sum_{\tau_i
      \in {\bf T}_m^{<}}\left(\floor{\frac{t-D}{T_i}}+1\right)\times C_i
    \leq \sum_{\tau_i
      \in {\bf T}_m^{<}}\left(\frac{t-D}{T_i}+1\right)\times C_i \nonumber\\
     \leq & \sum_{\tau_i
      \in {\bf T}_m^{<}} C_i + (t-D)u_i 
   \leq   D + (t-D) = t.
  \end{align}
  Moreover, with $\sum_{\tau_i \in {\bf T}_m} u_i \leq 1$, we know
  that when $t = H$
\begin{align}
    \label{eq:dbf-sum-multiple}
    \sum_{\tau_i \in {\bf T}_m}{\sc dbf}(\tau_i, H) 
   = &  \sum_{\tau_i
      \in {\bf T}_m^{<}}\left(\floor{\frac{H-D}{T_i}}+1\right)\times C_i + \sum_{\tau_i
      \in {\bf T}_m\setminus{\bf T}_m^{<}} \frac{H}{T_i}C_i\nonumber\\
   =_1 &  \sum_{\tau_i
      \in {\bf T}_m^{<}}\frac{H}{T_i} C_i + \sum_{\tau_i
      \in {\bf T}_m\setminus{\bf T}_m^{<}} \frac{H}{T_i}C_i
   =   H\left(\sum_{\tau_i
      \in {\bf T}_m} u_i\right) \leq H,\nonumber
  \end{align}
  where $=_1$ comes from the fact that $\frac{H}{T_i}$ is an
  integer for any $\tau_i$ in ${\bf T}_m^{<}$ and $T_i > D > 0$ so that
  $\floor{\frac{H-D}{T_i}}+1$ is equal to $\frac{H}{T_i}$.

  For any value $t > H$, the value of  $\sum_{\tau_i \in {\bf T}_m}{\sc
    dbf}(\tau_i, t)$ is equal to \\\mbox{$\sum_{\tau_i \in {\bf T}_m}{\sc
    dbf}(\tau_i, t-H) + \sum_{\tau_i \in {\bf T}_m}{\sc dbf}(\tau_i,
  H)$.} Therefore, we know that if $\sum_{\tau_i \in {\bf T}_m^{<}}
  \frac{C_i}{D} \leq 1$ and $\sum_{\tau_i \in {\bf T}_m} u_i \leq 1$,
  the task set ${\bf T}_m$ can be feasibly scheduled by EDF.
\end{proof}

\begin{theorem}
  \label{thm:hardness-equivalent}
  If there does not exist any APTAS for the \mbox{2D-DVP} problem, unless ${\cal
  P} = {\cal NP}$, there also does not exist any APTAS for the multiprocessor partitioned packing
  problem with constrained-deadline task sets.
\end{theorem}
\begin{proof}
  Clearly, the reduction in this section from the 2D-DVP problem to the
  multiprocessor partitioned packing problem with constrained
  deadlines is in polynomial time.

  For a task subset ${\bf T}'$ of ${\bf T}$, suppose that ${\bf
    V}({\bf T}')$ is the set of the corresponding vectors that are
  used to create the task subset ${\bf T}'$.  By
  Lemma~\ref{lemma:common-deadline}, the subset ${\bf T}_m$ of the
  constructed tasks can be feasibly scheduled by EDF on a processor if
  and only if $\sum_{\tau_i \in {\bf T}_m^{<}} C_i =
  \sum_{\tau_i \in {\bf V}({\bf T}_m)} v_{i,2}\leq 1$ and
  $\sum_{\tau_i \in {\bf T}_m} u_i = \sum_{\tau_i \in {\bf V}({\bf
      T}_m)} v_{i,1} \leq 1$.

  Therefore, it is clear that the above reduction is a perfect
  approximation preserving reduction. That is, an algorithm with a
  $\rho$ (asymptotic) approximation factor for the multiprocessor
  partitioned packing problem can easily lead to a $\rho$ (asymptotic)
  approximation factor for the 2D-DVP problem.
\end{proof}

\subsection{Hardness of the 2D-DVP problem}
\label{sec:2d-dvp-hard}

Based on Theorem~\ref{thm:hardness-equivalent}, we are going to show
that there does not exist APTAS for the 2D-DVP problem, which also
proves the non-existence of APTAS for the multiprocessor partitioned
packing problem with constrained deadlines.  

\begin{theorem}
  \label{thm:apx-hard-2-d}
  There does not exist any APTAS for the 2D-DVP problem, unless ${\cal
    P}={\cal NP}$.
\end{theorem}
\begin{proof}
  This is proved by an L-reduction, following a similar strategy in
  \cite{DBLP:journals/ipl/Woeginger97} by constructing an L-reduction
  from the Maximum Bounded 3-Dimensional Matching (MAX-3-DM), which is
  MAX SNP-complete \cite{Kann:1991:MBM:105391.105396}. Details are in the Appendix\ifbool{techreport}{,}{ in~\cite{packing-arxiv-2018},} where a short comment regarding an erroneous
  observation in~\cite{DBLP:journals/ipl/Woeginger97} is also
  provided.
\end{proof}
The following theorem results from 
Theorems~\ref{thm:hardness-equivalent} and \ref{thm:apx-hard-2-d}.

\begin{theorem}
  \label{thm:apx-hard}
  There does not exist any APTAS for the multiprocessor partitioned
  packing problem for constrained-deadline task sets, unless
  ${\cal P}={\cal NP}$.
\end{theorem}

\section{Concluding Remarks}
\label{sec:conclusion}

This paper studies the partitioned multiprocessor packing problem to
minimize the number of processors needed for multiprocessor partitioned scheduling.
Interestingly, there turns out to be a
huge difference (technically) in whether one is allowed faster
processors or additional processors. Our results are summarized in
Table~\ref{tab:summary}.
For general cases, the upper bound and lower bound for the first-fit
strategy in the deadline-monotonic partitioning algorithm are both
open. 
The focus of this paper is the multiprocessor partitioned packing
problem. If \emph{global scheduling} is allowed, in which a job can be
executed on different processors, the problem of minimizing the
number of processors has been also recently studied in a more general
setting by Chen et
al.~\cite{DBLP:conf/spaa/0009MS16,DBLP:conf/soda/ChenMS16} and Im et
al.~\cite{RTSS-Pruhs-2017}. They do not explore any periodicity of the
job arrival patterns.  Among them, the state-of-the-art online
competitive algorithm has an approximation factor (more precisely,
competitive factor) of $O(\log \log M)$ by Im et
al. \cite{RTSS-Pruhs-2017}.  These results are unfortunately not
applicable for the multiprocessor partitioned packing problem since
the jobs of a sporadic task may be executed on different processors.

\bibliographystyle{plainurl}
\bibliography{real-time}

\ifbool{techreport}{}{\end{document}}

\clearpage
\section*{Appendix}

\subsection*{Proofs related to Section~\ref{sec:deadline-monotonic}}

\begin{proofAppendix}{Theorem~\ref{thm:best-fit}}
We provide a task set that can be scheduled on two processors 
  but where Algorithm~\ref{algo:fitting} when applying the best-fit strategy uses $\frac{N}{2}$
  processors. Let $K \geq 4$ be an integer, $N$ is $2K$, and $H$ is sufficiently
  large, i.e., $H \gg K^K$.
  \begin{compactitem}
  \item Let $D_1=1$, $C_1=1/K$, and $T_1=H$.
  \item For  $i=2, 4, \ldots,2K$, let $D_i = K^{\frac{i}{2}-1}$,
    $C_i = K^{\frac{i}{2}-2}$, and $T_i=D_i$.
  \item For $i=3, 5, \ldots,2K-1$, let $D_i =
    K^{\frac{i-1}{2}}$, $C_i = K^{\frac{i-1}{2}} -
    K^{\frac{i-1}{2}-1}$, and $T_i=H$.
  \end{compactitem}
  
  Hence, in this input instance, $D_1=D_2=1$,
  $D_3=D_4=K$, $\cdots$, $D_{2K-1} = D_{2K}=K^{K-1}$. For the
  simplicity of presentation, we will omit any term
  multiplied with $1/H$ by assuming that this is positive and
  arbitrarily small.  When applying 
  DM partitioning,
  tasks $\tau_1$ and $\tau_2$ are both assigned on processor
  $1$. Then, we know that at time $t \geq 1$, ${\sc dbf}^*(\tau_1,
  t)+{\sc dbf}^*(\tau_2 , t) \approx \frac{1}{K} + \frac{t}{K}$. Clearly,
  $\tau_3, \tau_5, \tau_7, \ldots, \tau_{2K-1}$ are not eligible for
  processor $1$, because for $i=1,2,\ldots,K-1$ we have
\begin{align}
    & C_{2i+1} + {\sc dbf}^*(\tau_1, D_{2i+1})+{\sc dbf}^*(\tau_2 ,
    D_{2i+1}) \nonumber \\ \approx~ &K^{i}-K^{i-1} +\frac{1}{K} +\frac{K^i}{K} >
    K^i = D_{2i+1}.
 \end{align}

  Therefore, $\tau_3$ is assigned on processor $2$. When
  considering $\tau_4$, both processors are feasible, and processor
  $2$ has a higher approximate demand at time $D_4$, i.e., $ {\sc
    dbf}^*(\tau_1, D_4)+{\sc dbf}^*(\tau_2 , D_4) \approx \frac{1}{K}
  + \frac{K}{K}$ and $ {\sc dbf}^*(\tau_3, D_4) = C_3=K-1$. Therefore,
  $\tau_4$ is assigned on processor $2$ under the best-fit
  strategy. Similarly, $\tau_5, \tau_7, \ldots, \tau_{2K-1}$ are not
  eligible for processor $2$, because for $i=2,3,\ldots,K-1$ we have
 \begin{align}
& C_{2i+1} + {\sc dbf}^*(\tau_3, D_{2i+1})+{\sc dbf}^*(\tau_4 ,
D_{2i+1}) 
\nonumber \\ \approx~ & K^{i}-K^{i-1} + C_3 + \frac{K^i}{K} > K^i = D_{2i+1}.    
  \end{align}
  When considering $\tau_6$, the allocated three processors are all
  feasible, but processor $3$ has a higher approximate demand at time
  $D_6$. One can formally prove that task $\tau_{2i+1}$ is assigned to
  processor $i+1$ because $C_{2i+1} + {\sc dbf}^*(\tau_{2j+1},
  D_{2i+1})+{\sc dbf}^*(\tau_{2j+2} , D_{2i+1}) > D_{2i+1}$ for any
  $j=0,1,\ldots,i-1$. Moreover, since ${\sc dbf}^*(\tau_{2j+1},
  D_{2i+2})+{\sc dbf}^*(\tau_{2j+2} , D_{2i+2}) \approx C_{2j+1} +
  K^i/K < K^i - K^{i-1} = {\sc dbf}^*(\tau_{2i+1}, D_{2i+2})$ for any
  $1 \leq i \leq K-1$ and $j=0,1,\ldots,i-1$ due to the assumption $K
  \geq 4$, we know that processor $i+1$ has the highest approximate
  demand at time $D_{2i+2}$ among the first $i+1$ (allocated)
  processors. Thus, task $\tau_{2i+2}$ is assigned to processor
  $i+1$ due to the best-fit strategy.
  Therefore, we conclude that the best-fit strategy assigns $\tau_i$
  to processor $\ceiling{\frac{i}{2}}$. The resulting solution uses
  $K$ processors.

  Now, consider the following task assignment, in which $\tau_i$ is
  assigned on processor $1$ (resp., $2$) if $i$ is an odd (resp. even)
  number. Let ${\bf T}'_m$ be the set of tasks that are assigned on
  processor $m$. The assignment is feasible on processor $2$,
  as all the tasks are with implicit deadlines, and the total
  utilization is $100\%$. 
  The assignment is also feasible on processor $1$ by verifying the
  schedulability by using ${\sc dbf}$, i.e., the demand bound function without
  approximation! Since all tasks in ${\bf T}'_1$ have the same
  period, we only have to verify 
  ${\sc dbf}$ at time
  $1, K, K^2, K^3, \ldots K^{K-1}$, in which $\sum_{\tau_i \in {\bf T}'_1}
  {\sc dbf}(\tau_i, K^k) = K^k-1+\frac{1}{K}$ for $k=1,2,\ldots,K-1$.
   
  We will now show that when $t > K^{K-1}$, the ${\sc dbf}$ 
  of ${\bf T}'_1$ at time $t$ will still be no more than $t$, 
 i.e., showing that $\sum_{\tau_i \in {\bf T}'_1} {\sc dbf}(\tau_i, t)
    \leq t, \forall t > 0$.
  Since the $N/2=K$ tasks in ${\bf T}'_1$ have the same period, for the
  simplicity of presentation, let $T$ be $H$ with $T \gg K^K$.  We can divide the
  time interval $[0, \infty]$ into $[0, D_1), [D_1, D_3), \ldots,
  [D_{N-3}, D_{N-1}), [D_{N-1}, T), [T, T+D_1], [T+D_1, T+D_3],
  \ldots$. Suppose that $\ell$ is a non-negative integer and $j$ is an
  index $j \in \setof{1,3,5,\ldots, N-1,N+1}$, where $t$ is in interval
  $[\ell T + D_{j-2}, \ell T + D_j)$. Here, $D_{-1}$ is an auxiliary
  parameter set to $0$ and $D_{N+1}$ is an auxiliary parameter set to
  $T$ for brevity.

  Then, due to the
  parameters of task $\tau_i$ and $t \in [\ell T + D_{j-2}, \ell T +
  D_j)$, for task $\tau_i \in {\bf T}'_1$, we have $dbf(\tau_i, t) =
  (\ell+1) C_i \mbox{ if } i < j$ and $dbf(\tau_i, t) = \ell C_i
  \mbox{ if } j \leq i \leq N$. As a result, when $j \in
  \setof{5,7,\ldots,N-1}$ and $t \in [\ell T + D_{j-2}, \ell T +
  D_j)$, we have
  \begin{align*}
    \;\;& \sum_{\tau_i \in {\bf T}'_1} dbf(\tau_i, t) = \ell
    \sum_{\tau_i \in {\bf T}'_1} C_i + \sum_{\tau_i \in {\bf T}'_1
      \mbox{ and } i < j} C_i \\
\;\;= &\ell \left(\frac{1}{K}+ \sum_{i=1}^{K-1} K^i -
    K^{i-1}\right) + \left(\frac{1}{K}+ \sum_{i=1}^{(j-2-1)/2} K^i -
    K^{i-1}\right)\\
=\;\; &
\ell (K^{K-1}-1+\frac{1}{K}) + K^{(j-3)/2} - 1 + \frac{1}{K} \\
\leq \;\; &\ell T + D_{j-2}
  \end{align*}
  Moreover, when $j=1$, we have $\sum_{\tau_i \in {\bf T}'_1}
  dbf(\tau_i, t) = \ell (K^{K-1}-1+\frac{1}{K}) \leq \ell T$.  When
  $j=3$, we have $\sum_{\tau_i \in {\bf T}'_1} dbf(\tau_i, t) = \ell
  (K^{K-1}-1+\frac{1}{K}) + \frac{1}{K} \leq \ell T+D_1$.  When $j=N+1$, we
  have $\sum_{\tau_i \in {\bf T}'_1} dbf(\tau_i, t) = \ell
  (K^{K-1}-1+\frac{1}{K}) + K^{K-1}-1+\frac{1}{K} \leq \ell
  T+D_{N-1}$. Therefore, we reach the conclusion that $\sum_{\tau_i
    \in {\bf T}'_1} {\sc dbf}(\tau_i, t) \leq t, \forall t > 0$.


  Hence, there exists a feasible solution by using only $2$
  processors, but the DM 
  partitioning algorithm under BF
  uses $\frac{N}{2}$ processors. 
\end{proofAppendix}

\begin{proofAppendix}{Theorem~\ref{thm:worst-fit}}
  Suppose that $K$ is an integer, $N$ is $2K$, and $H$
  is sufficiently large, i.e., $H \gg K^K$.  
  Consider the following input task set:
  \begin{compactitem}
  \item Let $D_1=1$, $C_1=1$, and $T_1=H$.
  \item For  $i=2, 4, \ldots,2K$, let $D_i = K^{\frac{i}{2}}$,
    $C_i = K^{\frac{i}{2}-1}$, and $T_i=D_i$.
  \item For $i=3, 5, \ldots,2K-1$, let $D_i =
    K^{\frac{i-1}{2}}$, $C_i = K^{\frac{i-1}{2}} -
    K^{\frac{i-1}{2}-1}$, and $T_i=H$.
  \end{compactitem}

  We know that $D_2=D_3=K$, $D_4=D_5=K^2$, $\cdots$,
  $D_{2K-2} = D_{2K-1}$.  The proof is very similar to that of
  Theorem~\ref{thm:best-fit}. For the simplicity of presentation, we
  will omit any term multiplied with $1/H$ by assuming that
  this is positive and arbitrarily small.

  When applying DM partitioning, task $\tau_1$ and
  $\tau_2$ are both assigned on processor $1$.  One can formally prove
  that task $\tau_{2i+1}$ is assigned to processor $i+1$ because
  $C_{2i+1} + {\sc dbf}^*(\tau_{2j+1}, D_{2i+1})+{\sc
    dbf}^*(\tau_{2j+2} , D_{2i+1}) > D_{2i+1}$ for any
  $j=0,1,\ldots,i-1$. Moreover, since ${\sc dbf}^*(\tau_{2j+1},
  D_{2i+2})+{\sc dbf}^*(\tau_{2j+2} , D_{2i+2}) \approx C_{2j+1} +
  K^{i+1}/K > K^i - K^{i-1} \approx {\sc dbf}^*(\tau_{2i+1},
  D_{2i+2})$ for any $1 \leq i \leq K-1$ and $j=0,1,\ldots,i-1$, we
  know that processor $i+1$ has the lowest approximate demand at time
  $D_{2i+2}$ among the first $i+1$ (allocated) processors. Therefore,
  task $\tau_{2i+2}$ is assigned to processor $i+1$ due to the
  worst-fit strategy, and $N/2$ processors are allocated.

  Similarly, assigning $\tau_i$ on processor $1$ (resp., $2$) if $i$
  is an odd (resp. even) number is a feasible solution using two
  processors.  
\end{proofAppendix}

\begin{proofAppendix}{Theorem~\ref{thm:gamma-small-2}}
  Suppose that the system allocates the last processor when
  considering a certain task $\tau_\ell$.  
  If $\ell$ is $1$, 
  the  solution is optimal.  We consider $\ell \geq 2$.  That is, when
  considering $\tau_\ell$, for any $m \in \setof{1,2,\ldots,M}$,
  either the condition in Eq.~\eqref{eq:demand-bound-test} or
  Eq.~\ref{eq:utilization-bound-test} is
  violated. Algorithm~\ref{alg:dmp} hence uses $M+1$ processors and
  assigns task $\tau_\ell$ to that processor.

  The first $M$
  processors are categorized
  into two disjoint sets ${\bf M}_1$ and
  ${\bf M}_2$. 
  For any
  $m$ in ${\bf M}_1$, the condition in Eq.~\eqref{eq:demand-bound-test} is
  violated. For any $m$ in ${\bf M}_2$, 
  Eq.~\eqref{eq:demand-bound-test} holds 
  but 
  Eq.~\eqref{eq:utilization-bound-test} is
  violated. Hence,
\begin{align}
    C_\ell + \sum_{\tau_i \in {\bf T}_m}  {\sc dbf}^*(\tau_i, D_\ell) &> D_\ell, & \forall m \in {\bf M}_1 \label{eq:condition1-violated}\\
    u_\ell + \sum_{\tau_i \in {\bf T}_m}  u_i& > 1, & \forall m \in {\bf M}_2.\label{eq:condition2-violated}
  \end{align}
  If $|{\bf M}_1|$ is $0$, by Eq.~\eqref{eq:condition2-violated}, we have
  $\frac{\sum_{\tau_i \in {\bf T}} u_i}{M} > 1-u_\ell  \geq 1-\gamma$, in which
  the asymptotic approximation factor for this case is $\frac{1}{1-\gamma} \leq
  \frac{2}{1-\gamma}$.

  \emph{For the rest of the proof, we focus on the case that $|{\bf M}_1| > 0$}.
  Suppose that $|{\bf M}_2|$ is $x |{\bf M}_1|$, with $x \geq 0$ and $|{\bf
  M}_1|> 0$. That is, $|{\bf M}_1|=\frac{M}{1+x}$ and $|{\bf
  M}_2|=\frac{Mx}{1+x}$.
  To prove the approximation factor, we will build the lower bound of
  $\sum_{\tau_i \in {\bf T}} \frac{C_i}{T_i}$ and $\max_{t>0}\sum_{\tau_i \in
  {\bf T}} \frac{dbf(\tau_i, t)}{t}$.
  For notational brevity, we define the 
  two parameters $\beta$ and $k$: 
\begin{align}   
\beta &\overset{\mbox{def}}{\equiv}\; \frac{|{\bf M}_1|C_\ell}{\sum_{m \in {\bf
M}_1}\sum_{\tau_i \in {\bf T}_m} dbf(\tau_i,
D_{\ell})},\label{eq:def:multi-beta}\\
 k &\overset{{\mbox{def}}}{\equiv}\; \frac{ |{\bf M}_1|D_{\ell}-(1+\beta)
 \sum_{m \in {\bf M}_1}\sum_{\tau_i \in {\bf T}_m} dbf(\tau_i, D_{\ell})}{\sum_{m \in {\bf M}_1}\sum_{\tau_i \in {\bf
T}_m}dbf(\tau_i, D_{\ell})}.  \label{eq:def:multi-k2}
\end{align}
By definition, $\beta > 0$ and we also have
\begin{equation} 
 \frac{\sum_{m \in {\bf M}_1}\sum_{\tau_i \in {\bf T}_m}  {\sc dbf}(\tau_i, D_\ell)}{ |{\bf M}_1|D_\ell} = \frac{1}{1+k+\beta}.  \label{eq:dbf-bound-proof-1st}  
\end{equation}
Moreover, since ${\bf M}_1$ is not empty and $D_\ell > 0$, we also  have
$1+k+\beta > 0$. By (\ref{eq:dbf-bound-proof-1st}), we know that 
\begin{align}
  &\max_{t>0}\sum_{\tau_i \in {\bf T}} \frac{dbf(\tau_i, t)}{t} \geq  
  \sum_{\tau_i \in {\bf T}} \frac{dbf(\tau_i, D_\ell)}{D_\ell}\nonumber\\
  \geq &|{\bf M}_1|\frac{\sum_{m \in {\bf M}_1}\sum_{\tau_i \in {\bf
  T}_m}dbf(\tau_i, D_\ell)}{|{\bf M}_1|D_\ell} = \frac{M}{(1+x) (1+k+\beta)}.\label{eq:final-dbf-proof}
\end{align}

Based on
(\ref{eq:condition1-violated}), we know that
\begin{align}
  &  |{\bf M}_1|C_\ell + \sum_{m \in {\bf M}_1}\sum_{\tau_i \in {\bf T}_m}  {\sc dbf}^*(\tau_i, D_\ell) > |{\bf M}_1|D_\ell\nonumber\\
  \underset{D_i \leq D_\ell}{\Rightarrow} &  |{\bf M}_1|C_\ell + \sum_{m \in {\bf M}_1}\sum_{\tau_i \in {\bf T}_m}  {\sc dbf}(\tau_i, D_\ell) + \frac{(D_\ell-D_i)}{T_i}C_i> |{\bf M}_1|D_\ell\nonumber\\
  \underset{D_i \geq 0}{\Rightarrow} & 
  \sum_{m \in {\bf M}_1}\sum_{\tau_i \in {\bf T}_m}
\frac{D_\ell}{T_i}C_i > |{\bf M}_1|D_\ell - |{\bf M}_1|C_\ell   - \sum_{m \in
{\bf M}_1}\sum_{\tau_i \in {\bf T}_m}  {\sc dbf}(\tau_i, D_\ell)  \nonumber\\
 \underset{ (\ref{eq:def:multi-beta}),
(\ref{eq:def:multi-k2})}{\Rightarrow} & \sum_{m \in {\bf M}_1}\sum_{\tau_i \in {\bf T}_m}
\frac{D_\ell}{T_i}C_i > k \sum_{m \in {\bf M}_1}\sum_{\tau_i \in
  {\bf T}_m} {\sc dbf}(\tau_i, D_\ell)\nonumber\\
  \Rightarrow \;&   
\sum_{m \in {\bf M}_1}\sum_{\tau_i \in {\bf T}_m}
u_i > k \frac{\sum_{m \in {\bf M}_1}\sum_{\tau_i \in
  {\bf T}_m} {\sc dbf}(\tau_i, D_\ell)}{D_\ell}
  \underset{\eqref{eq:dbf-bound-proof-1st}}{=}\frac{M
  k}{(1+x)(1+k+\beta)}\label{eq:utilization-bound-proof-2nd}
\end{align}

By Eq.~\eqref{eq:condition2-violated},
Eq.~\eqref{eq:utilization-bound-proof-2nd}, and 
with
$\gamma \geq \frac{C_\ell}{\min\{T_\ell, D_\ell\}} \geq u_\ell$, we get 
\begin{align}\scriptsize
  \sum_{\tau_i \in {\bf T}} \frac{C_i}{T_i} \geq & \sum_{m \in {\bf M}_1}\sum_{\tau_i \in {\bf T}_m}
u_i + \sum_{m \in {\bf M}_2}\sum_{\tau_i \in {\bf T}_m}
u_i \nonumber\\
> &~(1-u_\ell)|{\bf M}_2| + \frac{M k}{(1+x)(1+k+\beta)}\nonumber\\
  \geq&~ M\left( \frac{x}{1+x} (1-\gamma) +
  \frac{k}{(1+x)(1+k+\beta)}\right)\label{eq:final-utilization-proof}
\end{align}

Any feasible solution to pack the tasks in ${\bf T}$ needs at least
$\sum_{\tau_i \in {\bf T}} \frac{C_i}{T_i}$ or $\max_{t>0}\sum_{\tau_i
  \in {\bf T}} \frac{dbf(\tau_i, t)}{t}$ processors. Thus, for
the lower bound of the number of required processors, by
Eq.~\eqref{eq:final-dbf-proof} and Eq.~\eqref{eq:final-utilization-proof},
\begin{align*}
    & \max\left\{\max_{t>0}\frac{\sum_{\tau_i \in {\bf T}} dbf(\tau_i, t)}{ t}, \sum_{\tau_i \in {\bf T}} \frac{C_i}{T_i} \right\}
    \\
    \underset{(\ref{eq:final-dbf-proof}),
      (\ref{eq:final-utilization-proof})}{\geq} \;\;\;& \frac{M}{1+x}\times\max
    \left\{\frac{1}{1+k+\beta}, x(1-\gamma)+\frac{k}{1+k+\beta}
    \right\}\\
    \geq_1 \qquad & M\times \frac{1-\gamma}{2+\beta-\gamma(1+k+\beta)}
       \geq_2 M\times\frac{1-\gamma}{2},
    \end{align*}
  where $\geq_1$ is because: 1) 
  $\frac{1}{1+k+\beta}$ is a constant with respect to $x$ and 
  $x(1-\gamma)+\frac{k}{1+k+\beta}$ is an increasing function with
  respect to $x$, \mbox{2) their} only intersection 
  happens when $x=\frac{1-k}{(1-\gamma)(1+k+\beta)}$, and 3) hence
  $\left\{\frac{M}{1+x}\times\max \left\{\frac{1}{1+k+\beta},
      x(1-\gamma)+\frac{k}{1+k+\beta} \right \} \right\} \geq
  \frac{M}{1+\frac{1-k}{(1-\gamma)(1+k+\beta)}}\frac{1}{1+k+\beta} = 
  M\times \frac{1-\gamma}{(1-\gamma)(1+k+\beta)+1-k} = M\times \frac{1-\gamma}{2+\beta-\gamma(1+k+\beta)}$.
The inequality $\geq_2$ is from the
  fact that $\frac{\beta}{1+k+\beta}$ is equal to $C_\ell/D_\ell$ by
  definition and is no more than $\max_{\tau_i \in {\bf
      T}}\frac{C_i}{\min\{T_i, D_i\}}$, defined as $\gamma$, i.e.,
  $\frac{\beta}{1+k+\beta} \leq \gamma$, which implies that
  $\beta-\gamma(1+k+\beta) \leq 0$.

    Hence, there must be at least $\frac{1-\gamma}{2}M$ processors
    in any feasible solution.
    Thus, the DM partitioning is an asymptotic
    $\frac{2}{1-\gamma}$-approximation algorithm for the multiprocessor
    partitioned packing problem.
\end{proofAppendix}

\subsection*{Proofs related to Section~\ref{sec:hardness}}

\begin{proofAppendix}{Theorem~\ref{theorem:gap-speedup-vs-processors} of the
property $\sum_{i=1}^{N} \frac{dbf(\tau_i, t)}{1+\epsilon} \leq t, \forall t
  > 0$}
Since the $N$ constructed tasks in the proof of
Theorem~\ref{theorem:gap-speedup-vs-processors} have the same period, for the
simplicity of presentation, let $T$ be
$\frac{(1+\epsilon)^{N-2}}{\epsilon^{N-1}}$. We can divide the time
interval $[0, \infty]$ into $[0, D_1), [D_1, D_2), \ldots, [D_{N-1},
D_N=T), [T+0, T+D_1), [T+D_1, T+D_2], \ldots$. Suppose that $\ell$ is
a non-negative integer and $j$ is an index $j \in \setof{1,2,3,\ldots,
  N}$, where $t$ is in interval $[\ell T + D_{j-1}, \ell T +
D_j)$. Here, $D_0$ is an auxiliary parameter set to $0$ for brevity.

Then, due to the
parameters of task $\tau_i$ and $t\in [\ell T + D_{j-1}, \ell T +
D_j)$, we have $dbf(\tau_i, t) = (\ell+1) C_i \mbox{ if } i < j$ and
$dbf(\tau_i, t) = \ell C_i \mbox{ if } j \leq i \leq N$. As a result,
when $j \in \setof{3,4,\ldots,N}$ and $t \in [\ell T + D_{j-1}, \ell T
+ D_j)$, we have
\begin{align*}
\;\;&  \sum_{i=1}^{N} \frac{dbf(\tau_i, t)}{1+\epsilon} = \left(\ell
  \sum_{i=1}^{N} \frac{C_i}{1+\epsilon} \right)+ 
  \sum_{i=1}^{j-1} \frac{C_i}{1+\epsilon} = \frac{\ell}{1+\epsilon} \left(1+\sum_{i=2}^{N} C_i\right) +\frac{1}{1+\epsilon} \left(1+\sum_{i=2}^{j-1} C_i\right) \\
 = \;\;&  \frac{\ell}{1+\epsilon} \left(1+\sum_{i=2}^{N} \frac{1}{\epsilon}\left(\frac{1+\epsilon}{\epsilon}\right)^{i-2} \right) + \frac{1}{1+\epsilon} \left(1+\sum_{i=2}^{j-1} \frac{1}{\epsilon}\left(\frac{1+\epsilon}{\epsilon}\right)^{i-2} \right) \\
 =_1\; &  \frac{\ell}{1+\epsilon} \left(1+\frac{1}{\epsilon}\left(\frac{(\frac{1+\epsilon}{\epsilon})^{N-2+1}-1}{\frac{1+\epsilon}{\epsilon}-1}\right)\right) 
+ \frac{1}{1+\epsilon} \left(1+\frac{1}{\epsilon}\left(\frac{(\frac{1+\epsilon}{\epsilon})^{j-1-2+1}-1}{\frac{1+\epsilon}{\epsilon}-1}\right)\right) \\
 = \;\;&  \frac{\ell}{1+\epsilon} \left(\frac{(1+\epsilon)^{N-1}}{\epsilon^{N-1}} \right) + \frac{1}{1+\epsilon} \left( \frac{(1+\epsilon)^{j-2}}{\epsilon^{j-2}}\right)\\
 = \;\; &\ell \left(\frac{(1+\epsilon)^{N-2}}{\epsilon^{N-1}} \right) + \left( \frac{(1+\epsilon)^{j-1-2}}{\epsilon^{j-1-1}}\right)\\
= \;\;& \ell T + D_{j-1} \leq t
\end{align*}
where $=_1$ is due to the  geometric sequence $C_2, C_3, \ldots, C_N$.

Similarly, when $j$ is $1$, we have $\sum_{i=1}^{N} \frac{dbf(\tau_i,
  t)}{1+\epsilon} = \ell T \leq t$ and when $j$ is $2$ we have
$\sum_{i=1}^{N} \frac{dbf(\tau_i, t)}{1+\epsilon} = \ell T +1\leq
t$. Therefore, we reach the conclusion that $\sum_{i=1}^{N}
\frac{dbf(\tau_i, t)}{1+\epsilon} \leq t, \forall t > 0$.
\end{proofAppendix}

\subsection*{Comments on the Error
  in~\cite{DBLP:journals/ipl/Woeginger97} regarding non-existence of
  an APTAS for the two-dimensional vector packing problem}

We also find that the proof for the non-existence of an APTAS
  for the two-dimensional vector packing problem is erroneous 
  in~\cite{DBLP:journals/ipl/Woeginger97}. By using the same terminologies in
  \cite{DBLP:journals/ipl/Woeginger97}, here we explain this below. The error
  comes from a mistake in Observation 4 in~\cite{DBLP:journals/ipl/Woeginger97}
  for the feasibility to pack any arbitrary three vectors into a unit-bin. The
  correct observation is that 3 vectors can only be arbitrarily packed into a
  unit-bin if at most two are from $T$. By putting $3$ vectors generated from
  $T$ in one unit-bin, the sum in the first dimension will exceed $1$.
  Therefore, for the only-if part in the proof of Lemma 5 in
  \cite{DBLP:journals/ipl/Woeginger97}, it may require more than
  $\frac{3q+|T|-4\alpha}{3}$ unit-bins for packing the remaining
  $3q+|T|-4\alpha$ vectors in $U$. As a result, the vectors should be created
  more carefully as we will show in Section~\ref{sec:2d-dvp-hard}. By scaling
  the first dimension by a factor $\frac{4}{3}$ in
   vectors in ${\bf V}$ and excluding the dummy vectors corresponding ${\bf W}$
   from ${\bf V}$ in Section~\ref{sec:2d-dvp-hard}, it can be shown that
  the statement in Lemma 5 in~\cite{DBLP:journals/ipl/Woeginger97} can hold, and
  the hardness property for the two-dimensional vector packing problem can be
  proved.

\subsection*{Proofs related to Section~\ref{sec:2d-dvp-hard}}

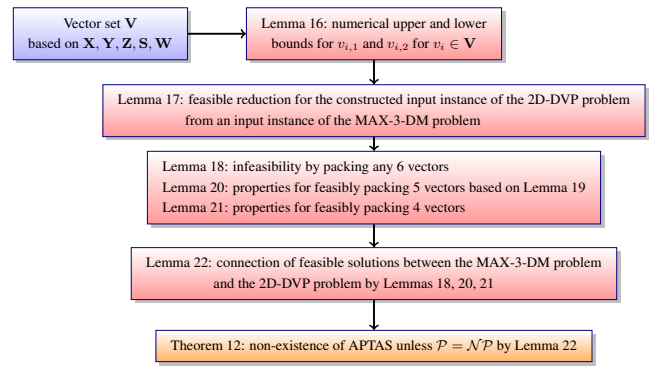
\begin{figure}[t]
  \centering
  \scalebox{0.60}{
  \begin{tikzpicture}[node distance = 1.5cm, auto]
    \node[end] (input) {
      \begin{tabular}{c}
        Vector set ${\bf V}$ \\based on ${\bf X, Y, Z, S, W}$\\
      \end{tabular}
    };
      \node[entity,right of=input,xshift=4.5cm] (get-values) {
      \begin{tabular}{c}
        Lemma~\ref{lemma:vector-values}: numerical upper and lower\\ bounds for 
        $v_{i,1}$ and $v_{i,2}$ for $v_i \in {\bf V}$
      \end{tabular}
    };
    \node[entity,below of=get-values,yshift=-0.2cm] (feasible-2D-DVP) {
      \begin{tabular}{l}
        Lemma~\ref{lemma:feasible-2D-DVP}: feasible reduction for the
        constructed input instance  of the 2D-DVP problem\\
        \hspace{1.4cm} from an input instance of the MAX-3-DM problem\\
      \end{tabular}
    };
    \node[entity,below of=feasible-2D-DVP,yshift=-0.2cm] (packing-properties) {
      \begin{tabular}{l}
        Lemma~\ref{lemma:6vectors}: infeasibility by packing any 6 vectors\\
        Lemma~\ref{lemma:5vectors}: properties for feasibly packing 5 vectors based on Lemma~\ref{lemma:5vectors-infeasibility}\\
        Lemma~\ref{lemma:4vectors}: properties for feasibly packing 4 vectors
      \end{tabular}
    };
    \node[entity,below of=packing-properties,yshift=-0.4cm] (max3dm-link) {
      \begin{tabular}{l}
        Lemma~\ref{lemma:connection}: connection of feasible solutions between the
        MAX-3-DM problem\\\hspace{1.4cm} and the 2D-DVP problem by Lemmas \ref{lemma:6vectors},~\ref{lemma:5vectors},~\ref{lemma:4vectors}
      \end{tabular}
    };
    \node[max,below of=max3dm-link,yshift=-0.1cm] (final) {
      \begin{tabular}{c}
        Theorem~\ref{thm:apx-hard-2-d}: non-existence of APTAS unless ${\cal P}={\cal NP}$ by Lemma~\ref{lemma:connection}
      \end{tabular}
    };
    \draw[very thick,->] (input) -- (get-values);
    \draw[very thick,->]  (get-values) -- (feasible-2D-DVP);
    \draw[very thick,->]  (feasible-2D-DVP) -- (packing-properties);
    \draw[very thick,->]  (packing-properties) -- (max3dm-link);
    \draw[very thick,->]  (max3dm-link) -- (final);
  \end{tikzpicture}}
\caption{The proof strategy for the non-existence of APTAS for the 2D-DVP problem.}
  \label{fig:detail-lemma-max-steps}
\end{figure}

Our proof strategy is shown in
Figure~\ref{fig:detail-lemma-max-steps}.
The first step of our L-reduction follows a similar strategy in
\cite{DBLP:journals/ipl/Woeginger97} by constructing an L-reduction
from the Maximum Bounded 3-Dimensional Matching (MAX-3-DM), which is
MAX SNP-complete \cite{Kann:1991:MBM:105391.105396}. 
The MAX-3-DM problem is defined as follows: We are given three sets
${\bf X} = \setof{x_1, \ldots, x_q}$, ${\bf Y} = \setof{y_1, \ldots,
  y_q}$, ${\bf Z} = \setof{z_1, \ldots, z_q}$ and a subset ${\bf S}
\subseteq {\bf X}\times{\bf Y}\times{\bf Z}$ so that each element in
${\bf X, Y, Z}$ occurs in one, two, or three triples in ${\bf S}$, i.e.,
$q \leq |{\bf S}| \leq 3q$. The goal is to find a maximum cardinality
subset ${\bf S}'$ of ${\bf S}$ such that no two triples in ${\bf S}'$
agree in any coordinate.

We denote the input instance for the MAX-3-DM problem by $I$ and the
optimal solution is with cardinality $\mbox{OPT}(I)$. In our proof, we
will use Observation 1 and Observation 2 from Woeginger in
\cite{DBLP:journals/ipl/Woeginger97}, restated here in
Lemma~\ref{lemma:opt-max3dm} and Lemma~\ref{lemma:4values}.


\begin{lemma}[Observation 1 from Woeginger
in~\cite{DBLP:journals/ipl/Woeginger97}]
  \label{lemma:opt-max3dm}
  The cardinality $\mbox{OPT}(I)$ of an optimal solution for any input instance $I$ of the
  MAX-3-DM problem is at least $\frac{q}{7}$.
\end{lemma}
%


For an input instance of the MAX-3-DM problem, let
\begin{align*}
&  x_i' = ir+1, & 1 \leq i \leq q,\\
&  y_i' = ir^2+2, & 1 \leq i \leq q,\\
&  z_i' = ir^3+4, & 1 \leq i \leq q,
\end{align*}
where $r=32q$.
For a triple $(x_i, y_j, z_k)$ in ${\bf S}$, we define
\[
s_\ell' = r^4-kr^3-jr^2-ir+8.
\]
Let ${\bf Q}$ be the set of the above $3q+|{\bf S}|$ integers, $x_i', y_i',
z_i', s_\ell'$. Moreover, let $b=r^4+15$. Resulting from this, we get:

\begin{lemma}[Observation 2 from Woeginger
in~\cite{DBLP:journals/ipl/Woeginger97}]
  \label{lemma:4values}
  Four integers in ${\bf Q}$ sum up to the value $b$ if and only if
  (1) one of them corresponds to some element $x_i \in {\bf X}$, one of them
  corresponds to some element $y_j \in {\bf Y}$, one of them corresponds to
  some element $z_k \in {\bf Z}$, and one of them corresponds to some triple
  $s_\ell \in {\bf S}$, and if (2) $s_\ell = (x_i, y_j, z_k)$ holds for
  these four elements.
\end{lemma}
\begin{proof}
  This property is the same as the Observation 2 in
  \cite{DBLP:journals/ipl/Woeginger97}. We provide a more
  comprehensive proof here.

  The if part is due to the definition of $s_\ell'$. The only-if
  part comes from the working modulo $r$, modulo $r^2$, modulo $r^3$,
  and modulo $r^4$, a slightly updated and changed version of an argument from \cite[page 98]{b:Gary79}, detailed as follows:
  \begin{compactitem}
  \item We denote the four integers in ${\bf Q}$ sum up to $b=r^4+15$ as $q_1, q_2, q_3, q_4$, i.e., $\sum_{h=1}^{4}q_h = r^4+15$. 
  \item By the definition of $x_i', y_i', z_i', s_\ell'$, we know that
    $(q_h \mbox{ modulo } r)$ is either $1,2,4,$ or $8$ for
    $h=1,2,3,4$. Moreover, $15 = (b \mbox{ modulo } r) = (\sum_{h=1}^4
    q_h) \mbox{ modulo } r= \sum_{h=1}^4 (q_h \mbox{ modulo } r)$,
    where the last equality is due to the facts that $r =32q \geq 32$
    and $(q_h \mbox{ modulo } r) \leq 8$. We can now enumerate all the
    combinations of the 4 values $q_1, q_2, q_3, q_4$. The only
    possibility to achieve $15=\sum_{h=1}^4 (q_h \mbox{ modulo } r)$
    is that each of the four integers $q_1, q_2, q_3,$ and $q_4$
    \emph{exactly} corresponds to one element in ${\bf X}, {\bf Y},
    {\bf Z}$, and ${\bf S}$, respectively. This proves the first part
    of the lemma.
  \item Therefore, without loss of generality, we consider that $q_1$
    is $x_i'$, $q_2$ is $y_j'$, $q_3$ is $z_k'$, and $q_4$ is
    $s_\ell'$ defined by a triple $(x_{i^*}, y_{j^*}, z_{k^*})$ in
    ${\bf S}$.  To prove the second part of the lemma, we need to show
    that $i^*$ equals to $i$, $j^*$ equals to $j$, and $k^*$ equals to
    $k$.
  \item We consider the modulo $r^2$. We have $(b \mbox{ modulo }
    r^2)=15$, $(q_1 \mbox{ modulo } r^2) = x_i'$, $(q_2 \mbox{ modulo
    } r^2) = 2$, $(q_3 \mbox{ modulo } r^2) = 4$, and $(q_4 \mbox{
      modulo } r^2) = r^2 - i^*r + 8$.\footnote{$(q_4 \mbox{
      modulo } r^2) = r^2 - i^*r + 8$ is due to the fact $-i^*r+8 < 0$
    since $i^* > 0$ and $r = 32q \geq 1$.} Therefore, $((x_i'+2+4+r^2 - i^*r+8)
    \mbox{ modulo } r^2) = 15$, which implies $((r^2 + ir - i^*r)
    \mbox{ modulo } r^2) = 0$.  Since $1 \leq i \leq q$, $1 \leq i^* \leq q$, and $r =
    32q$, we know that $i^* \neq i$ results in $((r^2 + ir - i^*r)
    \mbox{ modulo } r^2)  \neq 0$. Therefore, $i^*$ must be equal to $i$.
\item The modulo $r^3$ with the same step above ensures that $j^*$ is equal to $j$. 
\item The modulo $r^4$ with the same step above ensures that $k^*$ is equal to $k$. 
  \end{compactitem}
  We therefore reach the conclusion of the lemma.
\end{proof}


The integers in ${\bf Q}$ are defined as the same as in
\cite{DBLP:journals/ipl/Woeginger97}. \emph{However, the constructed
  (reduced) two-dimensional vectors have to be carefully designed to
  be a feasible input instance for the 2D-DVP problem, whereas the
  hardness remains.} 
  Therefore, the rest of the proof is
different from \cite{DBLP:journals/ipl/Woeginger97}.
 For illustrating the proof
strategy, Fig.~\ref{fig:detail-lemma-max-steps}
provides a short summary. The reduced input instance for the
2D-DVP problem is to first create $3q+|{\bf S}|$ two-dimensional
vectors as follows:
\begin{subequations}
\begin{align}
v_i &= (0.18+\frac{3x_i'}{4\cdot 5b}, 0.26-\frac{x_i'}{5b}), & 1 \leq i \leq q,\\
v_{i+q} &= (0.18+\frac{3y_i'}{4\cdot 5b}, 0.26-\frac{y_i'}{5b}), & 1 \leq i \leq q,\\
v_{i+2q} &= (0.18+\frac{3z_i'}{4\cdot 5b}, 0.26-\frac{z_i'}{5b}), & 1 \leq i \leq q,\\
v_{i+3q} &= (0.06+\frac{3s_i'}{4\cdot 5b}, 0.42-\frac{s_i'}{5b}), & 1 \leq i \leq |{\bf S}|.
\end{align}
\label{eq:constructions-of-v}
\end{subequations}

For the simplicity of presentation, we say that a vector $v_i$
corresponds to set ${\bf X}, {\bf Y}, {\bf Z}$, or ${\bf S}$ if the
vector is constructed according to an element in the corresponding
set. Moreover, we create additional $|{\bf S}|$ vectors that are
invariant. 
Accordingly, we say that these
vectors are corresponding to a dummy vector set ${\bf W}$, where $|{\bf
  W}| = |{\bf S}|$:
\begin{align*}
v_{i+3q+|S|} &= (0.25, 0), & 1 \leq i \leq |{\bf W}|.
\end{align*}
We use ${\bf V}$ to denote the set of
the reduced vectors that are constructed above. The following lemma
shows that the above construction makes the vectors corresponding to
${\bf X}, {\bf Y}$, and ${\bf Z}$ almost similar to each other and
different from the vectors corresponding to ${\bf S}$.

\begin{lemma}
  \label{lemma:vector-values}
  For a vector $v_i$ in ${\bf V}$,
  \begin{enumerate}
  \item $0.18 < v_{i,1} < 0.185$ and $0.25374 < v_{i,2} < 0.26$ if
    $v_i$ corresponds to ${\bf X}, {\bf Y}$, or ${\bf Z}$;
  \item $0.205 < v_{i,1} < 0.21$ and $0.22 < v_{i,2} < 0.2265$ if
    $v_i$ corresponds to ${\bf S}$.
  \end{enumerate}
\end{lemma}
\begin{proof}
By definitions with $q \geq 1$, we know that
\begin{align}
0<& \frac{x'_i}{5b} < \frac{q\cdot 32q+1}{5\times(32q)^4} < 0.0000063, \\
0<& \frac{y'_i}{5b} < \frac{q\cdot (32q)^2+2}{5\times(32q)^4} <  0.0002, \\
0<& \frac{z'_i}{5b} < \frac{q\cdot (32q)^3+4}{5\times(32q)^4} < 0.00626.
\end{align}
Moreover, we have
\begin{align}
  0.2 > \frac{s'_\ell}{5b} &>\frac{(32q)^4-q\cdot (32q)^3-q\cdot (32q)^2-q\cdot 32q+8}{5\times((32q)^4+15)} >0.1935.
\end{align}

Therefore, by taking the above inequalities and the definition of
vectors in ${\bf V}$, the statement in the lemma 
is simple arithmetic.
\end{proof}

\begin{lemma}
  \label{lemma:feasible-2D-DVP}
  The constructed input instance above from an input instance of the MAX-3-DM problem is a feasible input of the 2D-DVP problem
\end{lemma}
\begin{proof}
  By construction $v_{i,1} > 0$ for any constructed $v_i$.  Based on
  Lemma~\ref{lemma:vector-values}, we know that $v_{i,2} > v_{i,1}$ holds for a
  vector $v_i$ corresponding to set ${\bf X}, {\bf Y}, {\bf Z}$, or ${\bf S}$, 
  whereas for a vector $v_i$ corresponding to set ${\bf W}$ we know that
  $v_{i,2}=0$.   Moreover, 
  since $x_i', y_i', z_i, s_i', b$ are positive integers, 
  we also know that $v_{i,1}$ and $v_{i,2}$ are rational
  numbers by our constructions.  Hence, ${\bf V}$ is a
  feasible input instance for the 2D-DVP problem.
\end{proof}

Now, we can show the hardness due to the vector set
${\bf V}$.  
The following lemmas (Lemma~\ref{lemma:6vectors} to Lemma~\ref{lemma:4vectors})
are based on numerical properties 
of the construction of 
${\bf V}$, considering 
to pack $6$
vectors, $5$ vectors, and $4$ vectors into a bin.
\begin{lemma}
  \label{lemma:6vectors}
  Any six vectors in ${\bf V}$ cannot be feasibly packed into a
  bin.
\end{lemma}
\begin{proof} The first dimension is at least $0.18$ for each vector. 
\end{proof}

\begin{lemma}
  \label{lemma:5vectors-infeasibility}
  If five vectors in ${\bf V}$ can be feasibly packed into a bin, then the
  following three properties hold:
  \begin{enumerate}
  \item at most one vector corresponds to ${\bf W}$,
  \item at most three vectors correspond to ${\bf X}$, ${\bf Y}$, or ${\bf Z}$, and
  \item at most one vector corresponds to ${\bf S}$.
  \end{enumerate}
\end{lemma}
\begin{proof}

This follows from the numerical properties in Lemma~\ref{lemma:vector-values}.

  {\bf Property 1:} 
  Suppose 
  two
  vectors are from ${\bf W}$.  
  The other three vectors
  cannot exceed $0.5$ in the first dimension. However, as any vector
  $v_i$ corresponding to ${\bf X}$, ${\bf Y}$, ${\bf Z}$, or ${\bf S}$
  has $v_{i,1} > 0.18$, we reach a contradiction. For 3, 4, and 5 vectors in
  ${\bf W}$, the proof is identical. 

  {\bf Property 2:} 
  As the second dimension for
  any vector corresponding to ${\bf X}$, ${\bf Y}$, or ${\bf Z}$ is
  larger than $0.25$, if there are more than three vectors
  of these vectors, the second
  dimension will be more than $1$.

  {\bf Property 3:} 
  Suppose that there are $\ell \in \setof{2,3,4,5}$
  vectors from ${\bf S}$ for contradiction. 
  By property 1, we only have to
  consider whether a vector from ${\bf W}$ is one of the five vectors
  or not. Therefore, there are two sub-cases. (Sub-case 1:) If all the
  other $5-\ell$ vectors are corresponding to ${\bf X}$, ${\bf Y}$, or
  ${\bf Z}$, (i.e., none of them corresponds to ${\bf W}$), the sum in
  the second dimension is more than $0.25374\times
  (5-\ell)+0.22\times\ell>1.10$. (Sub-case 2:) If one vector
  corresponds to ${\bf W}$, and other $4-\ell$ vectors correspond to ${\bf
    X}$, ${\bf Y}$, or ${\bf Z}$, then the sum in the first dimension
  is more than $0.18\times (4-\ell)+0.205\times\ell+0.25\geq 1.02$ for $\ell=2,3,4$.
  Therefore, the third property holds.
\end{proof}

\begin{lemma}
  \label{lemma:5vectors}
  Five vectors in ${\bf V}$ can be feasibly packed into a bin if and only
  if (1) one of them corresponds to some element $x_i \in {\bf X}$, one of
  them corresponds to some element $y_j \in {\bf Y}$, one of them
  corresponds to some element $z_k \in {\bf Z}$, one of them corresponds to
  some element in ${\bf W}$, and one of them corresponds to some
  triple $s_\ell \in {\bf S}$, and (2) $s_\ell = (x_i, y_j, z_k)$ holds
  for 
  the elements from ${\bf X}, {\bf Y}, {\bf Z}, {\bf S}$.
\end{lemma}
\begin{proof}
  The if-part is based on the definition. We focus on the only-if
  part.  
   Based on
  Lemma~\ref{lemma:5vectors-infeasibility}, to feasibly pack $5$
  vectors, denoted 
  here as ${\bf V}'$,  
  three of them correspond to
  ${\bf X}, {\bf Y}, {\bf Z}$, one corresponds to ${\bf S}$, and one
  corresponds to ${\bf W}$.
  We know that 
  \mbox{$\sum_{v_i \in {\bf V}'} v_{i,1} \leq 1$} and
  $\sum_{v_i \in {\bf V}'} v_{i,2} \leq 1$.
  Let $\sigma$ be the sum of the \emph{four} integers in
  ${\bf Q}$ that are used to construct 
  the vectors from ${\bf X}, {\bf Y}, {\bf Z}, {\bf S}$ in ${\bf V}'$.

  As a result, we know that $\sum_{v_i \in {\bf V}'} v_{i,1} = 0.85 +
  \frac{3 \sigma}{4\cdot 5b}\leq 1$, which implies $\sigma \leq
  b$. Similarly, we have $\sum_{v_i \in {\bf V}'} v_{i,2} = 1.2 -
  \frac{\sigma}{5b}\leq 1$, which implies $\sigma \geq b$. Hence,
  $\sigma=b$ must hold. Therefore, the observation in Lemma~\ref{lemma:4values}
  yields the only-if part.
\end{proof}

\begin{lemma}
  \label{lemma:4vectors}
  Four vectors in ${\bf V}$ can be feasibly packed into a bin if 1) exactly three
  of them correspond to elements in ${\bf X}\cup {\bf Y}\cup{\bf Z}$
  and one of them corresponds to an element in ${\bf W}$, or 2) four
  of them correspond to elements in ${\bf S}\cup {\bf W}$.
\end{lemma}
\begin{proof}
  These properties are based on the numerical inequalities in
  Lemma~\ref{lemma:vector-values}. Let the ${\bf V}'$ be the set of
  the four vectors. For the first case, we have $\sum_{v_i \in {\bf
      V}'} v_{i,1} <0.185\times 3 + 0.25 = 0.805\leq 1$ and $\sum_{v_i
    \in {\bf V}'} v_{i,2} <0.26\times 3 + 0 =0.78\leq 1$.  For the
  second case, we have $\sum_{v_i \in {\bf V}'} v_{i,1} \leq 0.21
  \times (4-\ell) + 0.25\times \ell \leq 1$ and $\sum_{v_i \in {\bf
      V}'} v_{i,2} \leq 0.2265\times(4-\ell) + 0 \leq 1$, where $\ell$
  is the number of vectors in ${\bf V}'$ that corresponds to ${\bf
    W}$ and $0 \leq \ell \leq 4$.
\end{proof}

Based on the above lemmas, the following lemma provides a connection
between the feasible solutions of the MAX-3-DM problem and the
multiprocessor partitioned packing problem.

\begin{lemma}
  \label{lemma:connection}
  Let $\eta >0$ be an integer such that $\frac{3q+2|{\bf S}|-\eta}{4}$
  is an integer. There exists a feasible solution for the input
  instance $I$ of the MAX-3-DM problem that contains at least $\eta$
  triples if and only if there exists a feasible packing for reduced
  input instance ${\bf V}$ of the 2D-DVP problem
  that uses at most $\frac{3q+2|{\bf S}|-\eta}{4}$
  bins.
\end{lemma}
\begin{proof}
  {\bf only-if}: Let ${\bf S}'$ with $|{\bf S}'|=\eta$ be the feasible
  solution of the MAX-3-DM problem. Based on
  Lemma~\ref{lemma:5vectors}, we know that we can feasibly pack
  $5\eta$ vectors among the $3q+2|{\bf S}|$ vectors in ${\bf V}$ by
  using $\eta$ bins, in which $\eta$ vectors corresponding to ${\bf
    W}$, 
    $\eta$ vectors corresponding to ${\bf S}$,
  and $3\eta$ vectors corresponding to ${\bf X}$, ${\bf Y}$,
  and ${\bf Z}$ are chosen. Note that by definition $\eta$ is at most
  $q$.  For the remaining $3(q-\eta)$ vectors corresponding to ${\bf
    X}, {\bf Y}$, or ${\bf Z}$, we can group three of them by using
  $(q-\eta)$ bins. For each of these $(q-\eta)$ bins, based on
  Lemma~\ref{lemma:4vectors} and the fact that $|{\bf W}| = |{\bf S}| \geq q$, we
  can additionally assign one remaining vector corresponding to ${\bf
    W}$ such that these four vectors (one corresponding to ${\bf W}$,
  and three corresponding to ${\bf X}, {\bf Y}$, or ${\bf Z}$) can be
  feasibly packed in one bin. Again, based on
  Lemma~\ref{lemma:4vectors}, for the remaining $(|{\bf S}|-\eta)$
  vectors corresponding to ${\bf S}$ and $(|{\bf W}|-q = |{\bf S}|-q)$ vectors
  corresponding to ${\bf W}$, we can feasibly pack any four of them in
  a bin. Therefore, the above packing is feasible and requires exactly
  \[
  \frac{3q+2|{\bf S}|-5\eta}{4} + \eta =    \frac{3q+2|{\bf S}|-\eta}{4}
  \]
  bins, which is valid since $\frac{3q+2|{\bf S}|-\eta}{4}$ is assumed to be an integer.

  {\bf if}: Consider a feasible packing by using at most
  $\frac{3q+2|{\bf S}|-\eta}{4}$ bins. As there is no feasible packing
  for any $6$ vectors, a bin must have at most $5$ vectors. Suppose
  that exactly $\eta'$ bins are with $5$ vectors. If $\eta' < \eta$,
  the feasible packing requires at least $\frac{3q+2|{\bf
      S}|-5\eta'}{4}+\eta' = \frac{3q+2|{\bf S}|-\eta'}{4}$
  bins. Therefore, there must be at least $\eta$ bins with $5$
  vectors.  Then, based on Lemma~\ref{lemma:5vectors}, each of these
  $\eta$ bins is a corresponding triple in ${\bf S}$ for the MAX-3-DM problem, and no two triples agree in
  any coordinate, since each element in ${\bf X}, {\bf Y}, {\bf Z}$
  has only one corresponding vector in ${\bf T}$.
  Therefore, the
  input instance $I$ for the MAX-3-DM problem has a solution with at
  least $\eta$ triples and can be derived in polynomial time if a feasible 
  packing which uses at most $\frac{3q+2|{\bf S}|-\eta}{4}$ bins is given.
\end{proof}


\begin{proofAppendix}{Theorem~\ref{thm:apx-hard-2-d}}
  Consider any input instance $I$ for the MAX-3-DM problem.
  Suppose that there exists an APTAS, called Algorithm ${\cal A}$, for
  the 2D-DVP problem for contradiction.  We will show that this will
  contradict the MAX SNP-completeness of the MAX-3-DM problem \cite{Kann:1991:MBM:105391.105396}.  That is,
  unless ${\cal P} = {\cal NP}$, there does not exist any polynomial
  time approximation algorithm ${\cal A}'$ for the MAX-3-DM problem
  with
  \begin{equation}
    \label{eq:maxsnp}
    {\cal A}'(I) \geq (1-\epsilon)\mbox{OPT}(I),
  \end{equation}
  for an arbitrarily small real $\epsilon > 0$, where
  ${\cal A}'(I)$ is the number of triples in the solution derived from ${\cal A}'$.
    We also define
  $\delta$ as $\frac{\epsilon}{63}$.  Suppose that Algorithm ${\cal
    A}$ has an asymptotic guarantee to provide the following
  asymptotic approximation factor
\begin{equation}
    \label{eq:guarantee-A}
    {\cal A}({\bf V}) \leq (1+\delta)\mbox{OPT}({\bf V}) + \alpha^*,
  \end{equation}
  where ${\cal A}({\bf V})$ is the number of bins derived from
  Algorithm ${\cal A}$ for input instance ${\bf V}$, $\mbox{OPT}({\bf V})$ is the optimal solution
  of the 2D-DVP problem for input
  instance ${\bf V}$, and $\alpha^*$ is a constant. By definition,  ${\cal A}({\bf V})$ is a positive integer.

 If the maximum cardinality of a feasible ${\bf S}'$ is small,
  the input instance $I$ can be solved by checking all possible subsets
  of ${\bf S}$ with constant-bounded cardinalities. That is, if there
  does not exist any feasible solution ${\bf S}'\subseteq{\bf S}$ with
  $|{\bf S}'|=\frac{4}{\delta}(1+\delta+\alpha^*)$, checking all
  possible subsets ${\bf S}'$ of ${\bf S}$ with cardinality up to
  $\frac{4}{\delta}(1+\delta+\alpha^*)$ only takes polynomial time.
  Thus, if $\mbox{OPT}(I) < \frac{4}{\delta}(1+\delta+\alpha^*)$,
   the optimal solution of input instance $I$ can be determined
  in polynomial time.

  Now, we move to the remaining case that
\begin{equation}
    \label{eq:remaining-opt-large}
    \mbox{OPT}(I) \geq
    \frac{4}{\delta}(1+\delta+\alpha^*).
  \end{equation}
  The L-reduction by constructing ${\bf V}$ for the 2D-DVP problem 
  from the input instance $I$ of the MAX-3-DM problem can
  be done in polynomial time.
  Then, based on a result from Algorithm ${\cal A}$ to
  solve the input instance ${\bf V}$,
  we determine an integer $\eta$ with 
  \begin{equation}
    \label{eq:eta}
    \eta = 3q+2|{\bf S}|-4{\cal A}({\bf V}).
  \end{equation}
  Equivalent to (\ref{eq:eta}), we know
  \begin{equation}
    \label{eq:approximate-result}
    {\cal A}({\bf V}) = \frac{3q+2|{\bf S}|-\eta}{4}.
  \end{equation}
  Due to Lemma~\ref{lemma:connection}, we also 
  know that in the
  feasible solution derived from Algorithm ${\cal A}$, there must be
  at least $\eta$ bins with exactly five vectors in ${\bf V}$. We can
  construct a feasible solution ${\bf S}'$ for the input instance $I$ of
  the MAX-3-DM problem with cardinality equal to $\eta$, i.e., $|{\bf S}'|=\eta \geq
  \mbox{OPT}(I)$. According to the proof
  for the only-if part of Lemma~\ref{lemma:connection}, the
  construction of ${\bf S}'$ takes only polynomial time.
  Now, we will prove that such an ${\bf S}'$ 
  has an approximation factor of $1-\epsilon$ for the
  MAX-3-DM problem. By the facts that $\eta
  \geq \mbox{OPT}(I)$ and ${\cal A}({\bf V})=\frac{3q+2|{\bf S}|-\eta}{4}$ is an integer, we have
  \begin{equation}
    \label{eq:opt-upper}
    \mbox{OPT}({\bf V}) \leq {\cal A}({\bf V})  = \frac{3q+2|{\bf S}|-\eta}{4} \leq \ceiling{\frac{3q+2|{\bf S}|-\mbox{OPT}(I)}{4}}.
  \end{equation}
  By (\ref{eq:guarantee-A}), (\ref{eq:approximate-result}), and
  (\ref{eq:opt-upper}), we get
  \begin{align}
    {\cal A}({\bf V})   =\frac{3q+2|{\bf S}|-\eta}{4}  &= (1+\delta)
    \mbox{OPT}({\bf V})+ \alpha^* \nonumber\\&\leq
    (1+\delta)\left(\ceiling{\frac{3q+2|{\bf S}|-\mbox{OPT}(I)}{4}}\right) + \alpha^*\nonumber\\
    & \leq (1+\delta)\left(\frac{3q+2|{\bf S}|-\mbox{OPT}(I)}{4}+1\right) + \alpha^*    \label{eq:upper-A}
  \end{align}
  We re-organizing $\frac{3q+2|{\bf S}|-\eta}{4}\leq
  (1+\delta)\left(\frac{3q+2|{\bf S}|-\mbox{OPT}(I)}{4}+1\right) + \alpha^*$   
  in Eq.~\eqref{eq:upper-A}: 
  \begin{align*}
    (1+\delta) \mbox{OPT}(I) & \leq \eta+2|{\bf
    S}|\delta+3q\delta+4(1+\delta+\alpha^*)   \\ \leq_1 \eta+6q\delta+3q\delta+\delta \mbox{OPT}(I) & \leq_2 \eta+63\delta \mbox{OPT}(I)+\delta \mbox{OPT}(I),
  \end{align*}
  where $\leq_1$ comes from the definition that $|{\bf S}| \leq 3q$
  and Eq.~\eqref{eq:remaining-opt-large}, and $\leq_2$ comes from
  Lemma~\ref{lemma:opt-max3dm} as $\mbox{OPT}(I) \geq \frac{q}{7}$.
  Therefore, due to the setting of $\delta = \frac{\epsilon}{63}$, we
  reach the following inequality
  \begin{equation}
    \label{eq:final-ratio-1}
    \mbox{OPT}(I) \leq \eta+63\delta \mbox{OPT}(I) = \eta +\epsilon \mbox{OPT}(I).
  \end{equation}
  Now, we reach the approximation factor of the feasible solution
  ${\bf S}'$ for the input instance $I$ of the MAX-3-DM problem, in
  which
  \begin{equation}
    \label{eq:final-ratio-2}
    (1-\epsilon)\mbox{OPT}(I) \leq \eta = |{\bf S}'|.
  \end{equation}
  Hence, the MAX-3-DM problem, which is MAX-SNP-complete, can be solved in
  polynomial time with any approximation factor $1-\epsilon$ for any
  fixed $\epsilon$ with $0 < \epsilon < 1$. Therefore, this concludes
  that ${\cal P} = {\cal NP}$, which contradicts the assumption ${\cal
    P} \neq {\cal NP}$.
\end{proofAppendix}

\end{document}